%% file: arxiv_version.tex
\documentclass{article} % For LaTeX2e
\usepackage{iclr2027_conference,times}

\input{math_commands.tex}

\usepackage{hyperref}
\usepackage{url}

\iclrfinalcopy

\usepackage{algorithm}
\usepackage{algorithmic}
\usepackage{mathtools}
\usepackage{graphicx}
\usepackage{dsfont}
\usepackage{bm}
\usepackage{booktabs}

\usepackage{amsmath,amsthm,amssymb}
\usepackage{mathrsfs}
\usepackage{arydshln} % 加载虚线支持包
\usepackage{subcaption}
\usepackage{wrapfig}
\usepackage{multirow}
\usepackage{enumitem} % 控制itemize行宽
\usepackage[table]{xcolor} % 用于表格颜色
\usepackage[most]{tcolorbox}
\usepackage{makecell}

\usepackage{xcolor}
\usepackage{fancyvrb}
\usepackage{fvextra}

\tcbset{
  promptbox/.style={
    enhanced,
    breakable,
    colback=white,
    colframe=blue!35!black,
    colbacktitle=blue!15,
    coltitle=black,
    fonttitle=\bfseries,
    boxrule=1.2pt,
    arc=2mm,
    left=2mm,
    right=2mm,
    top=2mm,
    bottom=2mm,
  }
}

\theoremstyle{plain}
\newtheorem{theorem}{Theorem}[section]
\newtheorem{proposition}[theorem]{Proposition}
\newtheorem{lemma}[theorem]{Lemma}

\theoremstyle{definition}

\theoremstyle{remark}

\title{Error-Aware Reverse Auction Mechanism for Large Language Model Routing}

\author{
\textbf{Haolong Chen}$^{1,2,3,*}$, \textbf{Zhengyuan Xin}$^{2,3,*}$, \textbf{Liang Zhang}$^{4}$, \textbf{Lei Xue}$^{4}$, \textbf{Guangxu Zhu}$^{1,2,3,5}$ \\
$^1$Shenzhen International Center for Industrial and Applied Mathematics \\
$^2$Shenzhen Research Institute of Big Data \\
$^3$The Chinese University of Hong Kong, Shenzhen \\
$^4$Shenzhen Campus of Sun Yat-sen University \\
$^5$Shenzhen Loop Area Institute \\
$^*$Equal Contribution\\
\texttt{\{haolongchen1, zhengyuanxin\}@link.cuhk.edu.cn,} \\
\texttt{\{zhangliang27, xuelei3\}@mail.sysu.edu.cn,} \texttt{gxzhu@sribd.cn}
}

\begin{document}

\maketitle

% 把 ICLR 的 "Published as a conference paper..." 改掉
\fancyhead[L]{Preprint}

\begin{abstract}
Routing each query to a cost-effective large language model (LLM) is critical for balancing quality and cost, yet most routers rely on a centralized task center to predict model performance, creating an information-risk mismatch and a scalability bottleneck as the model pool grows. We formulate LLM routing as a market-based allocation problem among strategic providers and propose a routing paradigm that shifts ex-ante prediction to LLM providers via a reverse auction, where providers submit self-predicted acceptance probabilities and execution costs. To account for noisy provider predictions and center evaluations, we introduce the \textit{\textbf{E}rror-\textbf{A}ware \textbf{R}everse \textbf{A}uction \textbf{M}echanism} (EA-RAM), which explicitly models this Dual Error. We prove that, under a private-evaluation-belief structure, truthful effective-surplus reporting is incentive compatible in the reduced-form score space and individually rational under sellers' subjective beliefs, establish sufficient conditions for center rationality, and derive an explicit social-welfare loss bound. We further identify robustness effects: opposite-signed errors can cancel, vanishing-tail link functions (e.g., logistic) stabilize clear-cut cases via saturation, and extra noise smooths belief maps and reduces their maximal local sensitivity. Simulations and real-world benchmarks show that EA-RAM is robust to Dual Error and achieves a better cost--performance Pareto frontier than centralized baselines, with additional gains from provider-side local information, validating its practical effectiveness.
\end{abstract}

\input{section/intro}

\input{section/related}

\input{section/method}

\input{section/experiment}

\input{section/conclusion}

\subsection*{AI use statement}
In this work, we used LLMs to assist with code generation and language polishing.
We have reviewed all AI-assisted work.
We take responsibility for the final content of this work, including text, claims or artifacts produced with the aid of generative AI.

\subsection*{Reproducibility statement}
We provided a detailed description and implementation of our method in Section~\ref{sec:simulation_exp_setup}, Section~\ref{sec:realworld-setup}, and Appendix~\ref{app:implememtation_details}.

% (This section is \textbf{recommended} and does not count toward the page limit.)

% It is important that the work published in ICLR is reproducible. Authors are
% strongly encouraged to include a paragraph-long Reproducibility Statement at the
% end of the main text (before references) to discuss the efforts that have been
% made to ensure reproducibility. This paragraph should not itself describe
% details needed for reproducing the results, but rather reference the parts of
% the main paper, appendix, and supplemental materials that will help with
% reproducibility. For example, for novel models or algorithms, a link to an
% anonymous downloadable source code can be submitted as supplementary materials;
% for theoretical results, clear explanations of any assumptions and a complete
% proof of the claims can be included in the appendix; for any datasets used in
% the experiments, a complete description of the data processing steps can be
% provided in the supplementary materials. Each of the above are examples of
% things that can be referenced in the reproducibility statement.



\bibliography{iclr2027_conference}
\bibliographystyle{iclr2027_conference}

\input{section/appendix}

\end{document}

%% file: math_commands.tex
\usepackage{amsmath,amsfonts,bm}

\def\eqref#1{equation~\ref{#1}}
\def\1{\bm{1}}

\DeclareMathAlphabet{\mathsfit}{\encodingdefault}{\sfdefault}{m}{sl}
\SetMathAlphabet{\mathsfit}{bold}{\encodingdefault}{\sfdefault}{bx}{n}

\newcommand{\E}{\mathbb{E}}

\newcommand{\Var}{\mathrm{Var}}

%% file: section/intro.tex
\section{Introduction}
\label{sec:intro}

\begin{wrapfigure}{r}{0.53\textwidth}
    \vspace{-4.1em}
    \centering
    \includegraphics[width=\linewidth]{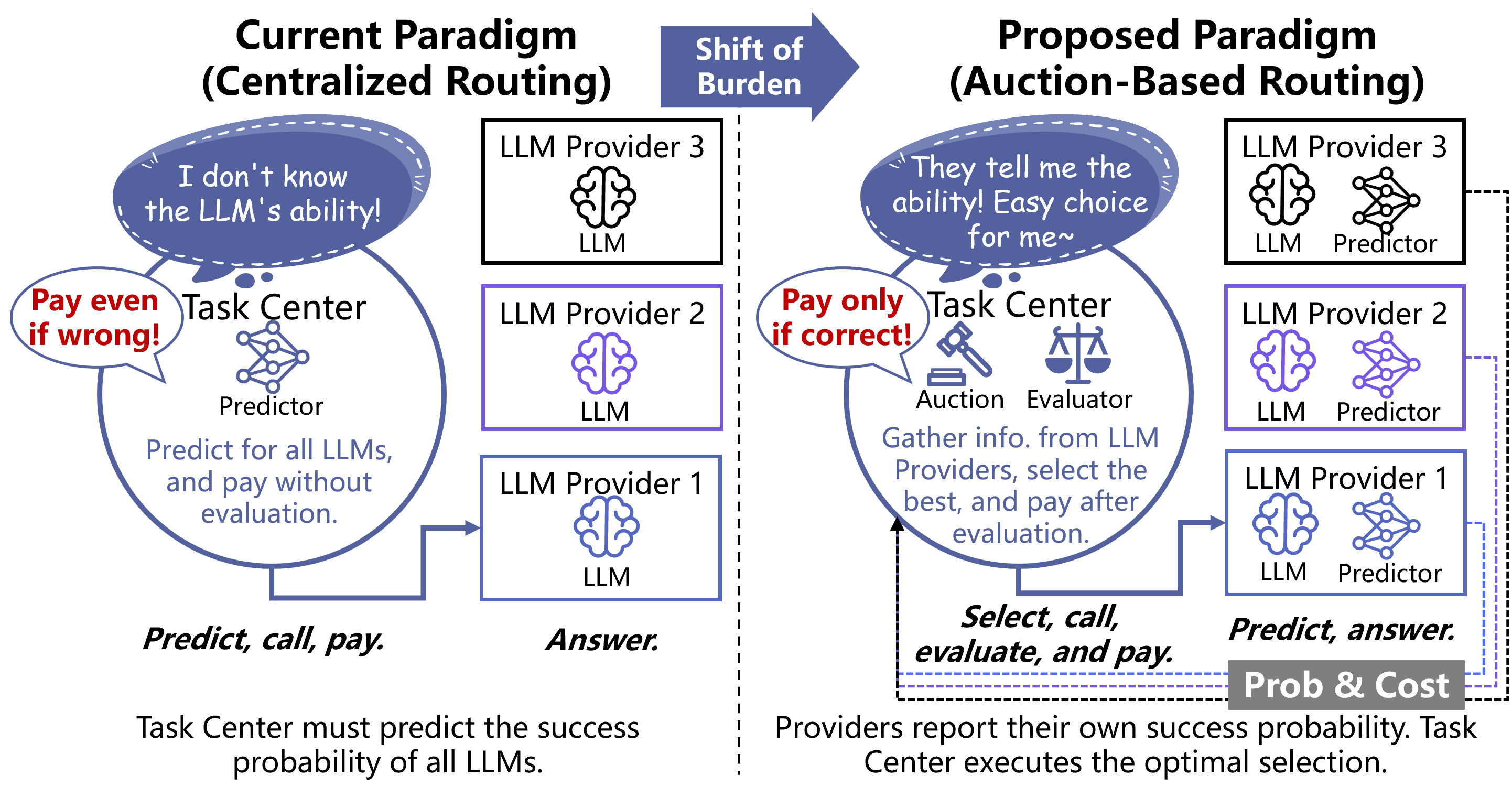}
    \vspace{-1.6em}
    \caption{Paradigm shift from traditional Centralized Routing to our proposed Auction-Based Routing.}
    \label{fig:paradigm_shift}
    \vspace{-1.5em}
\end{wrapfigure}

Large language models (LLMs) underpin many intelligent applications~\citep{guo2024large, chen2026overview, chen2026erskill, li2026qoereasoner}, yet the expanding model ecosystem makes it increasingly necessary to route each query to a cost-effective model. Because model capabilities and costs vary widely, from expensive frontier models to cheaper specialized ones, \emph{LLM routing} is essential for optimizing the cost--performance trade-off.

Most existing routers follow a centralized estimation paradigm, where the task center predicts each model's performance~\citep{zhuang2025embedllm, zhang2026router}. It has two structural limitations. \textit{(i) Information--risk mismatch:} the task center bears the risk of failure but has less information about the LLMs, whereas providers hold richer private knowledge. Commercial routers\footnote{e.g., \cite{openrouter_auto_router} and \cite{requesty_smart_routing}} largely retain this paradigm and thus inherit the same inefficiency. \textit{(ii) Scalability bottleneck:} the center must profile or train for every new model, incurring per-model overhead even for training-free routers~\citep{wang2025icl}, which hinders rapid expansion.

To address these, we propose a market-based routing paradigm that transfers ex-ante prediction to LLM providers via a reverse auction (Figure~\ref{fig:paradigm_shift}). The task center acts as the buyer and solicits bids from 
% providers, who report their self-predicted success probabilities and costs;
providers, who submit surplus scores formed from their local performance estimates and execution costs;
the center maintains only a model-agnostic ex-post evaluator. This distributed design aligns information with risk and removes the need for model-by-model profiling for the buyer, thereby improving scalability.

Similar paradigm shifts have proven successful in other mature domains. For example, moving from static allocation to market-based mechanisms such as real-time bidding in advertising~\citep{yuan2013real, wang2017display} and auction-based spectrum allocation~\citep{cramton1997fcc, milgrom2004putting}.
However, transferring this idea to LLM routing is \emph{not} straightforward. 
% Designing such a mechanism requires a new theory because LLM routing operates under an inherent \emph{Dual Error}: providers' ex-ante success estimates are subjective and noisy, and the center's ex-post evaluation is imperfect.
Analyzing such a mechanism requires extending classical execution-contingent
mechanism design to an inherent Dual-Error information structure:
providers' ex-ante beliefs are noisy, while the center's ex-post evaluation
is imperfect.
This setting departs from prior fault-tolerant allocation mechanisms~\citep{porter2008fault, takahashi2018strategic, bhatt2025coalesce}, which typically assume \emph{ideal observability}---i.e., task outcomes are objectively verifiable and/or agents' success probabilities are common knowledge---and thus can distort incentives if such assumptions are violated.

We therefore propose the \textit{\textbf{E}rror-\textbf{A}ware \textbf{R}everse \textbf{A}uction \textbf{M}echanism} (EA-RAM).
% \footnote{Code page: anonymous.4open.science/r/EARAM-8E0D}.
Crucially, EA-RAM explicitly models the inherent Dual Error by characterizing both providers' subjective error in ex-ante prediction and the task center's error in ex-post evaluation.
We characterize equilibrium bidding and prove that, under Dual Error and a private-evaluation-belief structure, EA-RAM admits dominant-strategy truthful effective-surplus reporting in the reduced-form score space and is individually rational (IR) under sellers' subjective beliefs; we further provide sufficient conditions for center rationality (CR) and derive an explicit upper bound on welfare loss relative to the error-free benchmark.
Beyond these, we uncover three structural robustness effects: ex-ante and ex-post errors with opposite signs can offset each other; when the link function has vanishing tails (e.g., logistic), large-margin (clear-cut) instances lie in the saturated region and are thus insensitive to noise; and
% extra independent noise smooths the belief maps, reducing their maximum slope and thereby limiting the gains from marginal manipulation.
extra noise smooths belief maps and reduces their maximal local sensitivity.
Extensive simulations and real-world experiments show that EA-RAM is robust to Dual Error and achieves a superior cost--performance trade-off over state-of-the-art centralized baselines, especially when leveraging providers' local information. Our contributions are summarized as follows:
\begin{enumerate}[itemsep=0em, topsep=0em, leftmargin=1em] 
    \item We propose EA-RAM, a market-based routing framework that shifts ex-ante prediction to LLM providers, addresses centralized routers' information--risk mismatch and per-model profiling bottleneck, and models LLM routing as an auction under explicit Dual Error.
    % \item We establish theoretical guarantees under Dual Error: EA-RAM satisfies BIC and IR, admits sufficient conditions for CR, and enjoys a welfare-loss bound; we further reveal robustness insights, including error compensation, saturation stability, and noise-induced flattening.
    \item EA-RAM is DSIC in the reduced-form effective-surplus space under private evaluation beliefs, satisfies IR, admits sufficient conditions for CR, and enjoys a welfare-loss bound; it further reveals error compensation, saturation stability, and noise-induced flattening.
    % \item Extensive experiments on simulations and real-world benchmarks show that EA-RAM is robust to Dual Error and improves economic efficiency compared to state-of-the-art centralized baselines.
    \item Extensive experiments on simulations and real-world benchmarks show that EA-RAM is robust to Dual Error and improves cost--performance trade-offs compared to centralized baselines.
\end{enumerate}

%% file: section/related.tex
\begin{table}[t]
\centering
\small
\caption{Conceptual comparison of standard VCG, FTMD, and EA-RAM.
EA-RAM retains an execution-contingent externality-payment structure
but differs in the information governing allocation, settlement, and welfare.}
\vspace{-0.6em}
\label{tab:mechanism-comparison}
\resizebox{\linewidth}{!}{
\begin{tabular}{lccc}
\toprule
 & \textbf{Standard VCG}~\citep{groves1973incentives} & \textbf{FTMD}~\citep{porter2008fault} & \textbf{EA-RAM} \\
\midrule
Modeled uncertainty
& Standard private values
& Execution uncertainty
& Dual Error \\

Decision-relevant value
& True private value
& $Vp_i-c_i$
& $Vg_i-c_i$ \\

Settlement basis
& Reported-value externality
& True execution outcome $\mu_i$
& Evaluator signal $\tilde{\mu}_i$ \\

Information alignment
& Valuation governs incentives and welfare
& $p_i$ governs incentives and welfare
& $g_i$, $h_i$, and $p_i$ govern
  incentives, settlement, and welfare \\

Incentive guarantee
& DSIC in value space
& DSIC in execution-surplus space
& DSIC in reduced-form
  effective-surplus space \\

Efficiency guarantee
& Exact
& Exact under observable execution outcomes
& Bounded welfare loss under Dual Error \\
\bottomrule
\end{tabular}}
\vspace{-1.5em}
\end{table}

\section{Related Work}
\label{sec:related_work}

% \paragraph{Market-Based Mechanisms for LLMs.}
% Auction-based methods have been explored for LLM content allocation \citep{hajiaghayi2024ad,duetting2024mechanism,dubey2024auctions} and task outsourcing \citep{bhatt2025coalesce, habibi2025edge, xu2026cached, fan2026cost}. 
% % Existing task-allocation mechanisms, following classical FTMD \citep{porter2008fault}, typically assume accurate provider predictions and reliable ex-post evaluation. In contrast, we study reverse-auction routing under \emph{Dual Error}, explicitly accounting for both provider-side prediction error and center-side evaluation error while establishing incentive and welfare guarantees.
% Regarding task allocation, prior reverse-auction mechanisms such as FTMD
% \citep{porter2008fault} address execution uncertainty using
% execution-contingent transfers. EA-RAM retains this classical
% externality-based payment structure, but studies a different information
% regime in which true fulfillment, center evaluation, and seller belief need
% not coincide. As summarized in Table~\ref{tab:mechanism-comparison}, this
% Dual-Error structure changes the strategic type, center-rationality
% conditions, and attainable welfare guarantees.

\paragraph{Market-Based Mechanisms for LLMs.}
Auction-based methods have been studied for LLM content allocation
\citep{hajiaghayi2024ad,duetting2024mechanism,dubey2024auctions}
and task outsourcing
\citep{bhatt2025coalesce,habibi2025edge,xu2026cached,fan2026cost}.
While mechanisms such as FTMD \citep{porter2008fault} handle execution
uncertainty via execution-contingent transfers, EA-RAM studies a
\emph{Dual-Error} regime where fulfillment, center evaluation, and seller
belief may differ, changing the strategic type, center rationality, and welfare
guarantees (Table~\ref{tab:mechanism-comparison}).

\paragraph{LLM Routing.}
Existing LLM routers mainly rely on centralized performance estimation, including predictor-based routing \citep{song-etal-2025-irt,ong2025routellm,zhuang2025embedllm,zhang2026agentrouter,zhang2026router, niu2026routing,zhang2026mtrouter}, in-context adaptation \citep{wang2025icl}, and cascading \citep{chen2024frugalgpt,dekoninck2025a}.
However, existing strategies predominantly rely on a \emph{centralized estimation paradigm}, which places the prediction burden on the router, the party with the least internal visibility, resulting in inherent information asymmetry and scalability bottlenecks.

%% file: section/method.tex
\begin{figure*}[t]
    \centering
    \includegraphics[width=1\linewidth]{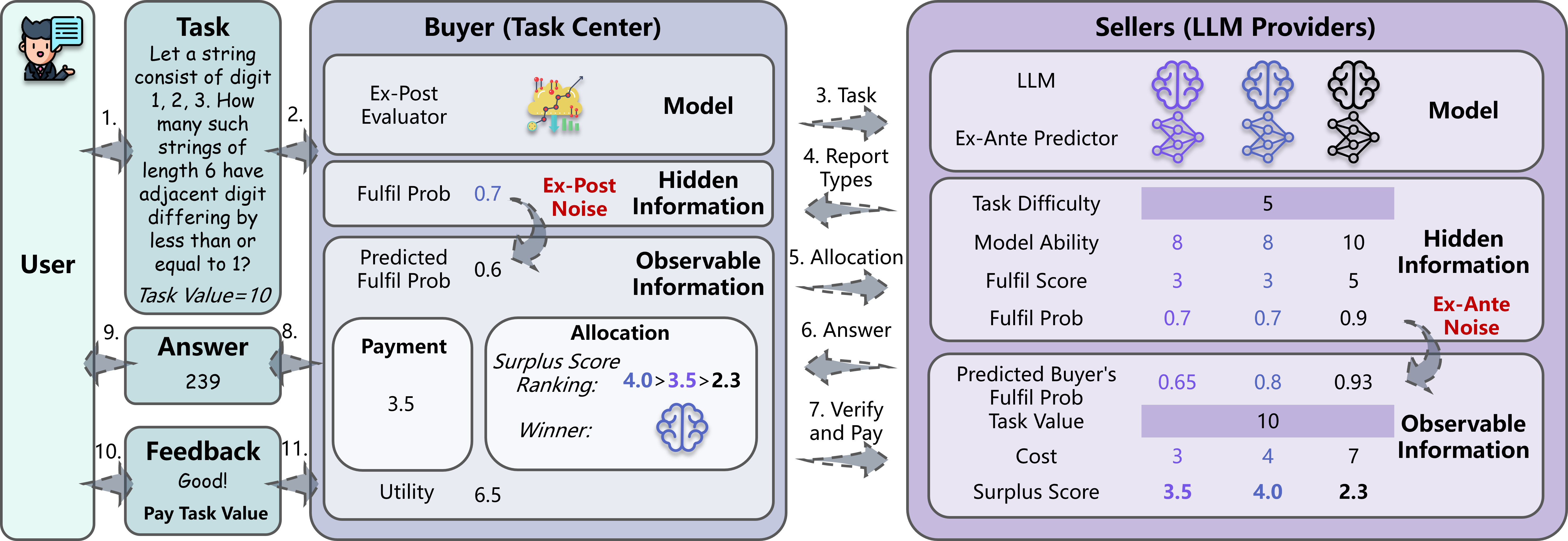}
    \vspace{-1.6em}
    \caption{EA-RAM for LLM routing.
    (1) Providers bid using ex-ante predictions.
    (2) The buyer allocates the query by reported surplus.
    (3) The selected model answers.
    (4) The buyer evaluates the output and settles payment.
    (5) The model output is returned to the user.
    }
    \label{fig:main}
    \vspace{-1.0em}
\end{figure*}

\section{Error-Aware Reverse Auction Mechanism}
\label{sec:error-aware-setup}

\subsection{Basic Setting}\label{sec:basic_setting}

% We formulate LLM routing as a strategic reverse auction between the task center and LLM providers. As shown in Figure~\ref{fig:main}, providers submit ex-ante bids, the center allocates the query, and payments are settled after ex-post evaluation.
% We consider a task center (the \textit{buyer}) and a set of risk-neutral\footnote{Risk neutrality is reasonable in repeated, high-volume platform settings, where providers optimize long-run expected profit and single-query risk is small relative to their portfolio.} LLM providers (the \textit{sellers}) $\mathcal{I}=\{1,\dots,N\}$ over independent, non-combinatorial \textit{tasks} $t\in\mathcal{T}$. Hence, we omit the superscript $t$. Each task has a common-knowledge \textit{value} $V>0$, realized by the buyer iff the task demand is fulfilled ($\mu=1$).
% Seller $i$ has private \textit{type} $\theta_i=(p_i,c_i)$, where $c_i>0$ is the \textit{execution cost} and $p_i$ is the true \textit{fulfillment probability} of the binary \textit{fulfillment indicator} $\mu_i\sim\mathrm{Bernoulli}(p_i)$.

We formulate LLM routing as a strategic reverse auction between the task center and LLM providers. As shown in Figure~\ref{fig:main}, providers submit ex-ante bids, the center allocates the query, and payments are settled after ex-post evaluation.
We consider a task center (the \textit{buyer}) and a set of risk-neutral\footnote{Risk neutrality is reasonable in repeated, high-volume platform settings, where providers optimize long-run expected profit and single-query risk is small relative to their portfolio.} LLM providers (the \textit{sellers}) $\mathcal{I}=\{1,\dots,N\}$ over independent, non-combinatorial \textit{tasks} $t\in\mathcal{T}$. Hence, we omit the superscript $t$. Each task has a common-knowledge \textit{value} $V>0$, realized by the buyer iff the task demand is fulfilled ($\mu=1$).
Seller $i$ is characterized by the physical primitives $(p_i,c_i)$, where $c_i>0$ is its execution cost and $p_i$ is the true fulfillment probability of the binary fulfillment indicator $\mu_i\sim\mathrm{Bernoulli}(p_i)$. These quantities characterize the underlying welfare environment and are distinct from the seller's strategic type introduced below.

\begin{wrapfigure}{r}{0.52\textwidth}
    \vspace{-1.4em}
    \centering
    \includegraphics[width=0.52\textwidth]{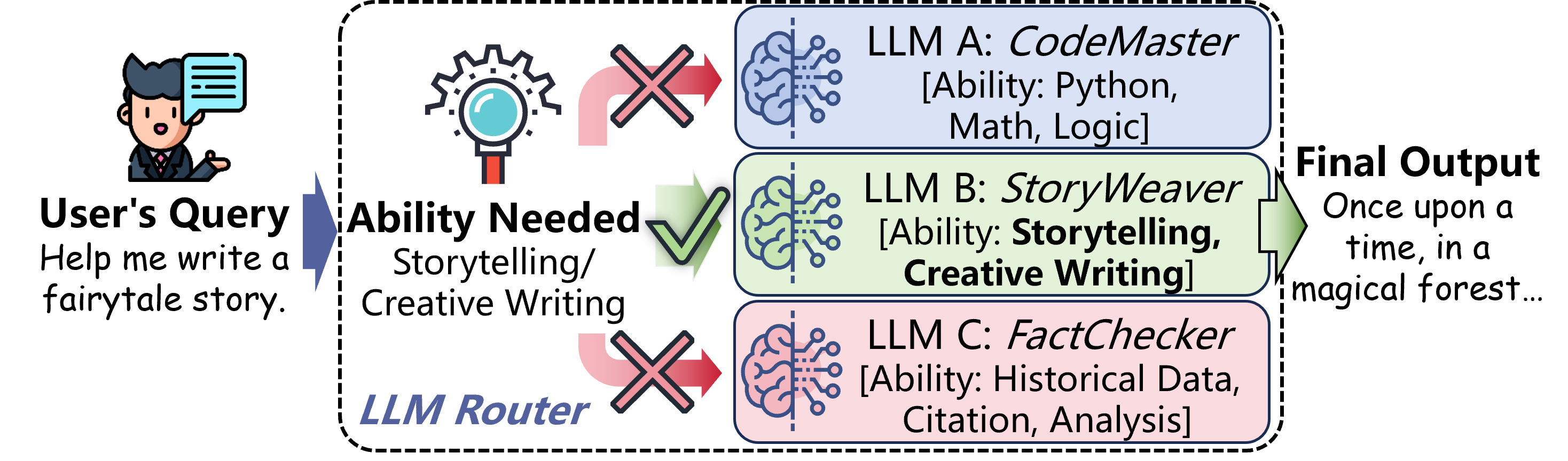}
    \vspace{-1.8em}
    \caption{Ability--Difficulty Matching.}
    \label{fig:difficulty_ability}
    \vspace{-1.2em}
\end{wrapfigure}

Inspired by capability modeling~\citep{wang2025icl, song-etal-2025-irt}, we model \textit{fulfillment probability} $p_i$ via the alignment between seller $i$'s \textit{model ability} $m_i\in\mathbb{R}$ and \textit{task difficulty} $d\in\mathbb{R}$ (Figure~\ref{fig:difficulty_ability}). The latent \textit{fulfillment score} $\phi_i=\phi(m_i,d)$ satisfies $\partial\phi/\partial m_i\ge0$ and $\partial\phi/\partial d\le0$, and $p_i=\sigma(\phi_i)$, where \textit{link function} $\sigma:\mathbb{R}\to(0,1)$ is strictly increasing, continuously differentiable, and globally $L_\sigma$-Lipschitz. Neither $p_i$ nor $\mu_i$ is observed during mechanism execution.

\subsection{Error-Aware Prediction \& Evaluation}
LLM Routing faces \textit{Dual Error} due to noisy prediction and imperfect evaluation, undermining classical mechanisms such as FTMD~\citep{porter2008fault} and necessitating an error-aware formulation.

\paragraph{Buyer's ex-post evaluation.}
% The buyer assesses the delivered output through an error-involved \textit{ex-post acceptance probability} $h_i=\sigma(\phi_i+\varepsilon_{\text{post}})$, where $\varepsilon_{\text{post}}$ represents a common \textit{evaluation error} characterized by mean $\E[\varepsilon_{\text{post}}]=a_{\text{post}}$ and variance $\Var(\varepsilon_{\text{post}})=b_{\text{post}}$. The buyer's final acceptance decision follows a Bernoulli draw $\tilde\mu_i \sim \mathrm{Bernoulli}(h_i)$, determining whether the answer is transmitted to the user.

% The buyer evaluates the output via the error-involved \textit{ex-post acceptance probability}
% $h_i=\sigma(\phi_i+\varepsilon_{\text{post}})$, where the \textit{evaluation error} $\varepsilon_{\text{post}}$ has
% $\E[\varepsilon_{\text{post}}]=a_{\text{post}}$ and
% $\Var(\varepsilon_{\text{post}})=b_{\text{post}}$.
% Decision $\tilde\mu_i\sim\mathrm{Bernoulli}(h_i)$ determines whether the answer is sent to the user.

% The buyer evaluates seller $i$'s output according to an evaluator whose criterion may deviate from the ground-truth fulfillment criterion.
% We model this evaluator-side distortion by $\epsilon_{\mathrm{post}}$, yielding the acceptance probability
% $h_i = \sigma(\phi_i + \epsilon_{\mathrm{post}})$,
% where
% $\mathbb{E}[\epsilon_{\mathrm{post}}] = a_{\mathrm{post}}$
% and
% $\mathrm{Var}(\epsilon_{\mathrm{post}}) = b_{\mathrm{post}}$.
% The resulting signal
% $\tilde{\mu}_i \sim \mathrm{Bernoulli}(h_i)$
% determines the buyer's evaluation outcome.

The buyer's evaluator may deviate from the ground-truth fulfillment
criterion. We model this evaluator-side uncertainty by
$\epsilon_{\mathrm{post}}$, with
$\E[\epsilon_{\mathrm{post}}]=a_{\mathrm{post}}$ and
$\Var(\epsilon_{\mathrm{post}})=b_{\mathrm{post}}$.
The buyer's marginal acceptance probability is
$h_i=\E_{\epsilon_{\mathrm{post}}}\!\left[\sigma(\phi_i+\epsilon_{\mathrm{post}})\right]$.
The resulting evaluation signal satisfies
$\tilde{\mu}_i\sim\mathrm{Bernoulli}(h_i)$.
The evaluator and settlement rule are publicly committed before bidding and
remain fixed throughout the auction. Thus, the buyer cannot alter the
evaluation rule after observing bids or model outputs.

\begin{wrapfigure}{r}{0.43\textwidth}
    \vspace{-2.8em}
    \centering
    \includegraphics[width=0.43\textwidth]{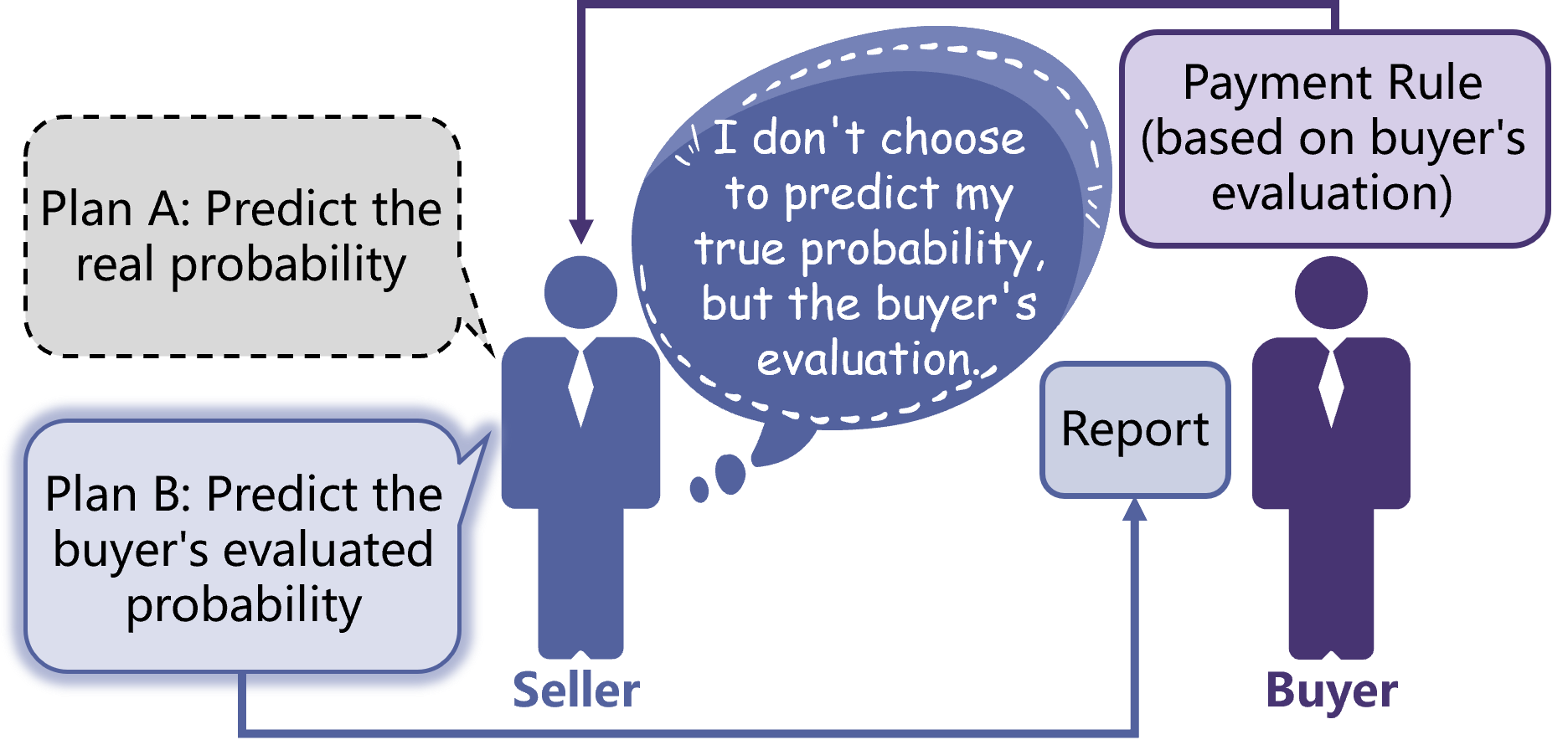}
    \vspace{-1.8em}
    \caption{Seller's strategic decision. The seller bids based on the buyer's error-involved evaluation, rather than the unobserved ground-truth.}
    \label{fig:seller_game}
    \vspace{-1.2em}
\end{wrapfigure}

\paragraph{Sellers' ex-ante prediction.}

Since payment depends on the buyer's evaluation, seller $i$ forms a
subjective belief about the buyer's acceptance criterion rather than the
unobserved ground-truth fulfillment probability. Let
$\epsilon_{\mathrm{ante},i}$ denote seller $i$'s belief error about the
buyer's evaluation criterion, independent of
$\epsilon_{\mathrm{post}}$, and define
$\eta_i=\epsilon_{\mathrm{post}}+\epsilon_{\mathrm{ante},i}$.
The seller's subjective acceptance probability is
$g_i=\E_{\eta_i}\!\left[\sigma(\phi_i+\eta_i)\right]$,
where
$\E[\eta_i]=a_{\mathrm{post}}+a_{\mathrm{ante},i}$ and
$\Var(\eta_i)=b_{\mathrm{post}}+b_{\mathrm{ante},i}$.
Thus, $\epsilon_{\mathrm{post}}$ captures the buyer's evaluation
distortion relative to ground truth, whereas
$\epsilon_{\mathrm{ante},i}$ captures seller $i$'s uncertainty about
that evaluation criterion.

\subsection{Mechanism Interaction}

% \begin{wrapfigure}[18]{r}{0.5\textwidth}
% % \vspace{-1.2em}
% \vspace{-2.5em} % arxiv
% \begin{minipage}{\linewidth}
% \begin{algorithm}[H]
% \fontsize{8.5pt}{9pt}\selectfont
% \caption{EA-RAM}
% \label{alg:mechanism}
% \begin{algorithmic}[1]
% \STATE \textbf{Input:} Sellers $\{m_i,c_i\}_{i=1}^N$; task $(d,V)$.
% \FOR{$i=1,\dots,N$}
%    \STATE Compute $(\phi_i,p_i)$ and $(h_i,g_i)$ with errors $\varepsilon_{\text{post}}$ and $\{\varepsilon_{\text{ante},i}\}_{i=1}^N$.
%    \STATE Set ranking score $\hat s_i \leftarrow V g_i - c_i$.
% \ENDFOR
% \STATE $j \leftarrow \arg\max_i \hat s_i$; $H \leftarrow \max(0,\max_{k\neq j}\hat s_k)$.
% \IF{$\hat s_j \le 0$}
%    \STATE \textbf{return} Null.
% \ENDIF
% \STATE Execute $\mu_j \sim \mathrm{Bern}(p_j)$; evaluate $\tilde{\mu}_j \sim \mathrm{Bern}(h_j)$.
% \STATE Pay $r_j \leftarrow V\tilde{\mu}_j - H$; set $r_i \leftarrow 0$ for all $i\neq j$.
% \STATE $U^{\text{seller}}_j \leftarrow r_j - c_j$; $U^{\text{buyer}} \leftarrow V\mu_j - r_j$.
% \STATE $W \leftarrow V \mu_j - c_j$.
% \STATE \textbf{return} $j$, $\{r_i\}_{i=1}^N$, $U^{\text{seller}}_j$, $U^{\text{buyer}}$, $W$.
% \end{algorithmic}
% \end{algorithm}
% \end{minipage}
% \vspace{-2.4em}
% \end{wrapfigure}

\begin{wrapfigure}[18]{r}{0.43\textwidth}
\vspace{-5.7em}
\begin{minipage}{\linewidth}
\begin{algorithm}[H]
\fontsize{8pt}{9pt}\selectfont
\caption{EA-RAM}
\label{alg:mechanism}
\begin{algorithmic}[1]
\STATE \textbf{Input:} Sellers $\mathcal{I}=\{1,\dots,N\}$; task $(d,V)$.
\FOR{$i=1,\dots,N$}
   \STATE Seller $i$ privately forms its evaluation belief $g_i$ and
   submits $\hat\theta_i=(\hat p_i,\hat c_i)$.
   \STATE Set the induced score $\hat s_i \leftarrow V\hat p_i-\hat c_i$.
\ENDFOR
\STATE $j \leftarrow \arg\max_i \hat s_i$; $H \leftarrow \max(0,\max_{k\neq j}\hat s_k)$.
\IF{$\hat s_j \le 0$}
   \STATE \textbf{return} Null.
\ENDIF
\STATE Seller $j$ executes the task and returns its output.
\STATE The buyer applies the committed evaluator and observes $\tilde{\mu}_j$.
\STATE Pay $r_j \leftarrow V\tilde{\mu}_j-H$; set $r_i \leftarrow 0$ for all $i\neq j$.
\STATE \textbf{return} $j$, $\{r_i\}_{i=1}^N$.
\end{algorithmic}
\end{algorithm}
\end{minipage}
\vspace{-2.4em}
\end{wrapfigure}

A \textit{mechanism} is a mapping
$\Gamma:\hat{\theta}\mapsto (f(\hat{\theta}), r(\hat{\theta},\tilde{\mu}))$,
where $f$ is the allocation rule and $r$ is the payment rule, based on sellers' reports $\hat{\theta}$ and the evaluator signal $\tilde{\mu}$.
Each seller submits a report $\hat{\theta}$; the buyer applies $f(\hat{\theta})$ to select a winner or the null allocation; after execution, the buyer observes $\tilde{\mu}$ and settles transfers via $r(\hat{\theta},\tilde{\mu})$.
Utilities and social welfare are then computed from the realized allocation, evaluation, and payment, as summarized in Algorithm~\ref{alg:mechanism}.
The latent variables $p_i$ and $\mu_i$ are used only for welfare analysis and need not be observed by the mechanism.

\par\WFclear
\paragraph{Utilities and welfare.}
% Our mechanism is designed with a welfare objective while recognizing that participants act strategically. Specifically, the buyer aims to maximize expected social welfare, whereas each seller chooses reports to maximize its own expected utility under the induced allocation and payment rules. Formally, any non-winning seller $i\ne j$ obtains zero realized and expected utility. The winning seller $j$ obtains \textit{realized utility} $U^{\text{seller}}_j=r_j-c_j$, where $r_j$ denotes the \textit{payment} for the winner and $c_j$ denotes the winner's cost. Accordingly, its \textit{expected utility} is $\mathbb E[U^{\text{seller}}_j]=\mathbb E[r_j]-c_j$, where the relevant expectation depends on the information available to the agent. The buyer's \textit{realized utility} is $U^{\text{buyer}}=V\mu_j-r_j$, with \textit{expected utility} $\mathbb E[U^{\text{buyer}}]=Vp_j-\mathbb E[r_j]$. Finally, \textit{realized social welfare} is $W=V\mu_j-c_j$, with \textit{expected welfare} $\mathbb E[W]=Vp_j-c_j$. Note that the realized welfare equals the sum of the buyer's and all sellers' realized utilities, formally described as $W = U^{\text{buyer}} + \sum_{i=1}^N U^{\text{seller}}_i$. This identity holds because transfers are internal to the mechanism and thus cancel out in aggregation: the buyer's payment $r_j$ is exactly the winner's received payment, so welfare depends only on the realized value $V\mu_j$ and the incurred cost $c_j$.

The buyer aims to maximize expected social welfare, while each seller maximizes its own expected utility under the mechanism. A non-winning seller receives zero utility, whereas the winner $j$ obtains $U^{\text{seller}}_j=r_j-c_j$, with $\mathbb{E}[U^{\text{seller}}_j]=\mathbb{E}[r_j]-c_j$. The buyer obtains $U^{\text{buyer}}=V\mu_j-r_j$, with $\mathbb{E}[U^{\text{buyer}}]=Vp_j-\mathbb{E}[r_j]$. Social welfare is $W=V\mu_j-c_j$, so $\mathbb{E}[W]=Vp_j-c_j$ and, since transfers are internal, $W=U^{\text{buyer}}+\sum_{i=1}^N U^{\text{seller}}_i$.

\paragraph{Reports.}

At the implementation level, seller $i$ submits
$\hat\theta_i=(\hat p_i,\hat c_i)$, inducing the score
$\hat s_i=V\hat p_i-\hat c_i$.
Since allocation and transfer depend on the report only through $\hat s_i$,
reports with the same score are outcome-equivalent, so we analyze incentives
in the reduced-form score space.
Under Dual Error, seller $i$'s strategic type is the effective surplus
$\bar T_i=Vg_i-c_i$.
Theorem~\ref{thm:bic} shows that any report satisfying
$V\hat p_i-\hat c_i=\bar T_i$ is weakly dominant in induced outcomes;
$(g_i,c_i)$ is one canonical representative, so the guarantee applies to
effective-surplus implementation rather than coordinate-wise identification.

\paragraph{Allocation.}
Given reports $\hat{\theta}$, the buyer allocates to maximize reported expected welfare. For seller $i$ with report $\hat{\theta}_i=(\hat p_i,\hat c_i)$, the reported surplus is $V\hat p_i-\hat c_i=\hat s_i$; hence ranking by $\hat s_i$ is equivalent to ranking sellers by reported expected welfare. The rule selects the largest positive $\hat s_i$, or the null outcome if all reported surpluses are non-positive.
Equivalently, introduce a dummy seller $0$ with $(\hat p_0,\hat c_0)=(0,0)$ and $\hat s_0=0$, representing the null allocation. Let $\bar{\mathcal I}=\mathcal I\cup\{0\}$ and $j\in\arg\max_{i\in\bar{\mathcal I}}\hat s_i$. Then $f(\hat{\theta})=j$ if $j\neq0$, and $f(\hat{\theta})=\emptyset$ otherwise.

\paragraph{Payment and incentive alignment.}
Let $j$ denote the winning seller and let $H=\max(0,\max_{k\neq j}\hat s_k)$ be the runner-up score. If payments depended on the self-reported probability $\hat p_i$, sellers could inflate $\hat p_i$ (e.g., report $\hat p_i=1$). We therefore condition transfers on the evaluator signal $\tilde\mu_j\in\{0,1\}$, a noisy proxy for the unobserved ground truth $\mu_j$. The winner is paid
$r_j=r(\hat\theta_j,\tilde\mu_j)=V\tilde\mu_j-H$~\footnote{In practice, negative realized transfers can be handled via a prefunded reserve maintained by each participating seller. This implementation-layer safeguard does not change the mechanism analyzed below.},
while losers receive $r_i=0$ for $i\ne j$. Here $V\tilde\mu_j$ rewards evaluated performance, and $H$ charges the winner for the externality imposed on others, aligning the best response with surplus-based truthful competition (see Theorem~\ref{thm:bic} for the formal statement).

\section{Mechanism Properties} \label{sec:properties}
This section establishes EA-RAM's theoretical foundations, showing that it remains incentive-aligned and economically robust under Dual Error. Proofs of all results are deferred to Appendix~\ref{appendix:proof}.

\subsection{Mechanism Goals}
To ensure a reliable and sustainable routing market, we follow FTMD~\citep{porter2008fault} and consider four standard desiderata:
\textbf{(1) Dominant-Strategy Incentive Compatibility (DSIC):} Truth-telling maximizes each seller's expected utility regardless of others' actions.
\textbf{(2) Individual Rationality (IR):} Truthful participation gives each seller non-negative expected utility.
\textbf{(3) Center Rationality (CR):} The buyer's expected utility remains non-negative.
\textbf{(4) Economic Efficiency (EE):} Under truthful reporting, the mechanism allocates to maximize total expected social welfare.

We formalize the \textit{error-aware setting} as a unified framework covering two cases. The general noisy case is the \textit{error-involved setting}; the ideal case with vanishing errors, $\varepsilon_{\text{post}}=0$ and $\varepsilon_{\text{ante},i}=0$, is the \textit{error-free setting}, corresponding to FTMD~\citep{porter2008fault}. In the error-free setting, ex-ante and ex-post beliefs coincide with the ground-truth probability, $g_i=h_i=p_i$, and the effective surplus report becomes the true surplus, $\hat s_i=\bar T_i=Vp_i-c_i$. We then establish the benchmark property:

\begin{proposition}[Optimality of the Error-Free Setting]
\label{prop:error_free_properties}
In the error-free setting, the mechanism satisfies all four desirable properties: DSIC, IR, CR, and EE.
\end{proposition}
This establishes the error-free benchmark, showing optimality under idealized conditions.
% Although Dual Error weakens these exact guarantees, \textbf{EA-RAM retains strong robustness in the error-involved setting: it satisfies dominant-strategy incentive compatible in effective-surplus space and IR, admits sufficient CR conditions when surplus margins absorb evaluation noise or evaluations are conservative, and bounds welfare loss}, ensuring controlled efficiency degradation under uncertainty.
Under Dual Error, exact welfare efficiency may be lost. Under the private-evaluation-belief structure introduced below, EA-RAM remains dominant-strategy incentive compatible in effective-surplus space and individually rational, admits sufficient conditions for CR, and enjoys a bounded welfare loss.

\subsection{Strategic Properties}
\label{sec:participant-rationalies}
% Let $H$ denote the runner-up score with cumulative distribution function $F_H$ and probability density function $f_H$, and let $\mathcal{I}_i$ denote seller $i$'s information at bidding time. We adopt a private-evaluation-belief structure: conditional on $\mathcal{I}_i$, competitors' reports do not alter seller $i$'s belief about the buyer's evaluation of its own output, i.e., $\mathbb{E}_i[\tilde{\mu}_i \mid H,\mathcal{I}_i]=\mathbb{E}_i[\tilde{\mu}_i \mid \mathcal{I}_i]=g_i$, where $\mathbb{E}_i[\cdot]$ denotes expectation under seller $i$'s subjective belief. Therefore, for a unilateral report $\hat{s}_i$, seller $i$'s interim expected utility is $U^{\mathrm{seller}}_i(\hat{s}_i)=\Pr(H\leq \hat{s}_i)\cdot\mathbb{E}_i[V\tilde{\mu}_i-H-c_i\mid H\leq \hat{s}_i]=F_H(\hat{s}_i)\bigl(\bar{T}_i-\mathbb{E}_i[H\mid H\leq \hat{s}_i]\bigr)$, where $\bar{T}_i=Vg_i-c_i$.

Let $H$ denote the runner-up score and let $\mathcal{I}_i$ denote seller $i$'s information at bidding time. We adopt the following private-evaluation-belief structure: conditional on $\mathcal{I}_i$, competitors' reports do not alter seller $i$'s belief about the buyer's evaluation of its own output, i.e., $\mathbb{E}_i[\tilde{\mu}_i \mid H,\mathcal{I}_i] = \mathbb{E}_i[\tilde{\mu}_i \mid \mathcal{I}_i] = g_i$.
Accordingly, for a unilateral report $\hat s_i$, seller $i$'s interim
subjective expected utility is
$U_i^{\mathrm{seller}}(\hat s_i)
=\E_i[(\bar T_i-H)\mathbf{1}\{H\leq\hat s_i\}\mid\mathcal I_i]$,
where $\bar T_i=Vg_i-c_i$.

% \begin{theorem}[Bayesian Incentive Compatibility (BIC)]
% \label{thm:bic}
% Seller $i$'s interim expected utility $U^{\text{seller}}_i(\hat s_i)$ is maximized at $\hat s_i=\bar T_i$. Hence, reporting $\hat s_i=\bar T_i$ is a Bayesian best response. Moreover, if $\bar T_i>0$, then $U^{\text{seller}}_i(\hat s_i)$ is uniquely maximized at $\hat s_i=\bar T_i$.
% \end{theorem}
% This implies that for every seller $i$, truthfully reporting $\hat s_i=\bar T_i$ is a Bayesian best response.
% Consequently, the strategy profile in which all sellers report $\hat s_i=\bar T_i$ constitutes a Bayesian Nash equilibrium. Under this report, seller $i$'s expected utility is
% $\E[U^{\text{seller}}_i] = \int_{-\infty}^{\bar T_i} (\bar T_i-h)\,dF_H(h)$.

% \begin{theorem}[Bayesian Incentive Compatibility under Subjective Evaluation Beliefs]
% \label{thm:bic}
% For seller $i$, let $\bar T_i=Vg_i-c_i$ denote its effective surplus under its subjective belief about the buyer's evaluation. Then seller $i$'s interim expected utility is maximized by reporting $\hat s_i=\bar T_i$. Hence, truthful effective-surplus reporting is a Bayesian best response, and the profile $\hat s=\bar T$ constitutes a Bayesian Nash equilibrium.
% \end{theorem}
% Thus, each seller is incentivized to align its reported score with its belief-adjusted effective surplus rather than with the unobserved ground-truth surplus. Under this report, seller $i$'s interim expected utility is $\E_i[U^{\mathrm{seller}}_i]=\E_i[(\bar T_i-H)\mathbf{1}\{H\leq \bar T_i\}]\geq 0$.

\begin{theorem}[Dominant-Strategy Incentive Compatibility in Effective-Surplus Space]
\label{thm:bic}
Under the private-evaluation-belief structure above, for seller $i$ with effective-surplus type $\bar{T}_i = Vg_i-c_i$, reporting $\hat{s}_i=\bar{T}_i$ maximizes its subjective expected utility for every realization of the runner-up score $H$. Hence, truthful effective-surplus reporting is a weakly dominant strategy in the reduced-form mechanism.
\end{theorem}
Thus, each seller is incentivized to align its reported score with its belief-adjusted effective surplus rather than with the unobserved ground-truth surplus. Under this report, seller $i$'s interim expected utility is $\E_i[U^{\mathrm{seller}}_i]=\E_i[(\bar T_i-H)\mathbf{1}\{H\leq \bar T_i\}]\geq 0$.

% \begin{theorem}[Individual Rationality (IR)]
% \label{thm:IR}
% At the equilibrium $\hat s=\bar T$, every seller satisfies IR:
% $\E[U^{\text{seller}}_i\mid \tilde\theta_i]\ge0$.
% \end{theorem}
\begin{theorem}[Individual Rationality (IR)]
\label{thm:IR}
Under truthful effective-surplus reporting $\hat s_i=\bar T_i$, every
seller satisfies interim IR under its subjective belief:
$\E_i[U_i^{\mathrm{seller}}\mid\mathcal I_i]\geq0$.
\end{theorem}
This implies that truthful reporting yields non-negative expected utility, thereby ensuring that participating in the auction remains a rational strategy for every seller.

\paragraph{Robustness to correlated beliefs.}
The private-evaluation-belief condition above is exact and may be violated
when competitors' scores reveal residual information about the common task.
We therefore quantify how the incentive guarantee degrades under bounded
dependence.

\begin{proposition}[Approximate IC under Correlated Evaluation Beliefs]
\label{prop:appro_ic}
Suppose seller $i$'s conditional evaluation belief satisfies
$|\E_i[\tilde\mu_i\mid H,\mathcal I_i]-g_i|\leq\delta_i$
for every runner-up score $H$ in its support. Then reporting
$\hat s_i=\bar T_i$ is $V\delta_i$-approximately dominant in the
reduced-form mechanism: for every alternative report and every $H$,
the conditional expected utility gain from deviating is at most
$V\delta_i$. Exact DSIC is recovered when $\delta_i=0$.
\end{proposition}

Thus, residual information contained in competitors' scores does not
immediately destroy incentive alignment: its maximal effect scales linearly
with the induced belief deviation $\delta_i$.

\begin{theorem}[Center Rationality (CR)]
\label{thm:CR-error-unified}
Each of the following sufficient conditions guarantees CR:
(A) $\E[H] \ge V\Delta_{\text{gate}}$, where $\Delta_{\text{gate}} = L_\sigma \sqrt{b_{\text{post}}+a_{\text{post}}^2}$;
(B) $\Delta_{\text{cons}}=\E\!\left[p_{(1)}-h_{(1)}\right]\ge 0$.
\end{theorem}
Intuitively, CR holds either (A) when the runner-up margin effectively buffers against evaluation noise, or (B) when the center's evaluation is sufficiently conservative. Under either condition, the mechanism is safe for the Task Center in the sense that its expected utility remains non-negative.

% \begin{proposition}[Comparative Statics: Ability and Difficulty]
% \label{prop:comparative_statics_unified}
% Holding fixed the runner-up score distribution $F_H$, seller $i$'s interim expected utility $\E[U^{\text{seller}}_i]$ is weakly increasing in their model ability $m_i$ and weakly decreasing in the task difficulty $d$.
% \end{proposition}
\begin{proposition}[Comparative Statics: Ability and Difficulty]
\label{prop:comparative_statics_unified}
Holding fixed the distribution of $H$, seller $i$'s interim subjective
expected utility is weakly increasing in model ability $m_i$ and weakly
decreasing in task difficulty $d$.
\end{proposition}
Intuitively, stronger models yield higher expected payoffs because they are more likely to meet the evaluator's standard, while increased task difficulty reduces potential gains.

\subsection{Welfare Analysis}
\label{sec:welfare-loss}

% For welfare accounting, let the expected welfare under a specified selected seller $i$ be $\E[W_i]=Vp_i-c_i$. We define the \textit{true optimal winner index} (selected under the error-free setting) as $i^\star \in \arg\max_j \{Vp_i-c_i\}$ and the \textit{error-aware setting's winner index} as $i^\dagger \in \arg\max_j\{V g_j-c_j\}$.

For welfare accounting, let the expected welfare under a specified selected seller $i$ be $\E[W_i]=Vp_i-c_i$. We extend the seller set to $\bar{\mathcal I}=\mathcal I\cup\{0\}$, where seller $0$ represents the null allocation with $p_0=g_0=h_0=c_0=0$. We define the \textit{true optimal allocation index} under the error-free setting as $i^\star\in\arg\max_{j\in\bar{\mathcal I}}\{Vp_j-c_j\}$, and the \textit{error-aware allocation index} as $i^\dagger\in\arg\max_{j\in\bar{\mathcal I}}\{Vg_j-c_j\}$.

% \begin{proposition}[Economic Efficiency (EE)]
% \label{prop:EE}
% The equilibrium allocation induced by the error-aware setting attains the same expected welfare as the error-free benchmark if and only if the error-aware winner $i^\dagger$ is welfare-optimal.
% \end{proposition}
\begin{proposition}[Economic Efficiency (EE)]
\label{prop:EE}
The equilibrium allocation induced by the error-aware setting attains the same expected welfare as the error-free benchmark if and only if the error-aware allocation index $i^\dagger$ is welfare-optimal.
\end{proposition}
Intuitively, Dual Error may distort the induced allocation away from the welfare-optimal seller, thereby causing welfare loss. We quantify this effect using the Dual Error's \emph{second-moment radii}, $M_{\mathrm{post}} = \sqrt{b_{\mathrm{post}}+a_{\mathrm{post}}^2}$ and $M_{\mathrm{ante}} = \max_i\sqrt{b_{\mathrm{ante},i}+a_{\mathrm{ante},i}^2}$.

% \begin{theorem}[Welfare-loss bound]
% \label{thm:welfare-loss-unified}
% The expected welfare loss satisfies
% $0 \le \E[W_{i^\star}]  -\E[W_{i^\dagger}] 
% = \E[(Vp_{i^\star}{-}c_{i^\star}) - (Vp_{i^\dagger}{-}c_{i^\dagger})] 
% \le 2V L_\sigma(M_{\mathrm{post}}+M_{\mathrm{ante}})$.
% \end{theorem}
% This bound provides a guarantee of robustness, showing that the efficiency degradation scales linearly with the aggregate magnitude of Dual Error and vanishes as the estimation quality improves.

\begin{theorem}[Welfare-loss bound]
\label{thm:welfare-loss-unified}
Let $i^\star\in\arg\max_{i\in\bar{\mathcal I}}\{Vp_i-c_i\}$ and $i^\dagger\in\arg\max_{i\in\bar{\mathcal I}}\{Vg_i-c_i\}$. Then the expected welfare loss satisfies $0\leq \E[W_{i^\star}]-\E[W_{i^\dagger}]\leq 2VL_\sigma(M_{\mathrm{post}}+M_{\mathrm{ante}})$.
\end{theorem}
This bound provides a guarantee of robustness, showing that the efficiency degradation scales linearly with the aggregate magnitude of Dual Error and vanishes as the estimation quality improves.

\subsection{Structural Insights}\label{sec:comp-statics}

Having established EA-RAM's validity under Dual Error, we now examine its structural insights.

\begin{proposition}[Opposite-Signed Error Compensation]
\label{prop:counteracting-error}
If $(g_i-h_i)(h_i-p_i)\le 0$, for $i\in\{i^\star,i^\dagger\}$, then the welfare-loss bound tightens to $2V L_\sigma\max\{M_{\mathrm{ante}},M_{\mathrm{post}}\}$.
\end{proposition}
Intuitively, when ex-ante and ex-post errors have opposite signs, they partially cancel out, reducing the resulting welfare loss.

\begin{proposition}[Saturation Robustness]
\label{prop:saturation-robustness}
In addition to the baseline link-function assumptions, suppose $\sigma$ has vanishing tails, i.e., $\lim_{|x|\to\infty}\sigma'(x)=0$. For any error $\epsilon$ with $\E[\epsilon^2]<\infty$, the deviation $\Delta(\phi_i)=|\E[\sigma(\phi_i+\epsilon)]-\sigma(\phi_i)|$ satisfies $\lim_{|\phi_i|\to\infty}\Delta(\phi_i)=0$.
\end{proposition}
Intuitively, flat-tailed links such as the logistic function saturate when $|\phi_i|$ is large, so score noise barely affects $\sigma(\phi_i)$ and the mechanism is robust in clear-cut cases.

\begin{proposition}[Noise-Induced Flattening]
\label{prop:variance_flattening}
Let $\eta$ be an independent noise term with $\E[|\eta|]<\infty$. Define the perturbed belief maps by convolution:
$\tilde g_i(\phi)=\E[g_i(\phi+\eta)]$ and $\tilde h_i(\phi)=\E[h_i(\phi+\eta)]$.
If $g_i$ and $h_i$ are continuously differentiable with bounded derivatives, then
$\sup_{\phi}|\tfrac{\partial \tilde g_i}{\partial \phi}(\phi)|\le \sup_{\phi}|\tfrac{\partial g_i}{\partial \phi}(\phi)|$
and
$\sup_{\phi}|\tfrac{\partial \tilde h_i}{\partial \phi}(\phi)|\le \sup_{\phi}|\tfrac{\partial h_i}{\partial \phi}(\phi)|$.
Consequently, injecting additional independent noise weakly reduces the maximal sensitivity of both $\phi\mapsto g_i(\phi)$ and $\phi\mapsto h_i(\phi)$.
\end{proposition}
% Intuitively, ex-ante error and ex-post error lower the maximal sensitivity of the ranking map and the approval map, respectively, limiting the return to strategic manipulation.
Intuitively, ex-ante error and ex-post error lower the maximal sensitivity of the ranking map and the approval map, respectively, reducing local sensitivity to strategic perturbations.

%% file: section/experiment.tex
\section{Experiments}
\label{sec:experiments}

% We evaluate EA-RAM through both controlled simulations and real-world benchmarks. Our experiments are designed to answer the following research questions:
% \begin{itemize}[itemsep=0em, topsep=0em, leftmargin=1em]
%     \item \textbf{RQ1 (Robustness):} How does EA-RAM perform under varying degrees of the \textit{Dual Error}?
%     \item \textbf{RQ2 (Economic Efficiency):} Can EA-RAM achieve a better cost--performance trade-off than traditional centralized strategies in real-world benchmarks?
%     \item \textbf{RQ3 (Information Utilization):} To what extent does the incorporation of seller-side local information gain enhance the overall routing quality of EA-RAM?
%     \item \textbf{RQ4 (Scalability and Efficiency):} How well does EA-RAM scale as the number of candidate LLM providers grows, and what are the resulting computation and communication costs?
% \end{itemize}

\subsection{Simulation Experiments}

\paragraph{Experimental Setup.}
\label{sec:simulation_exp_setup}

\begin{wraptable}{r}{0.23\linewidth}
\vspace{-6.6em}
\caption{Simulation Settings.}
\vspace{-0.7em}
\centering
\resizebox{\linewidth}{!}{
\setlength{\tabcolsep}{1.0pt}
\fontsize{6pt}{7pt}\selectfont
\begin{tabular}{ccccc}
\toprule
\textbf{Seller} & $m_i$ & $c_i$ & $p_i$ & $s_i$ \\ 
\midrule
0 & 1.0 & 9.0  & 0.38 & -1.45 \\
1 & 1.5 & 10.0 & 0.50 & 0.00 \\
2 & 2.5 & 12.0 & 0.73 & 2.62 \\
3 & 3.5 & 12.8 & 0.88 & 4.81 \\
4 & 4.5 & 13.8 & 0.95 & 5.25 \\ 
\bottomrule
\end{tabular}
}
\label{tab:sim_settings}
\vspace{-1em}
\end{wraptable}

% We simulate a routing environment with $N=5$ heterogeneous sellers competing for a task ($V{=}20, d{=}1.5$) (Table~\ref{tab:sim_settings}). 
% To test robustness, we construct a highly competitive market where the surplus gap between the top two sellers is small ($\Delta s = 0.44$); thus, even mild noise can flip their ranking.
% We evaluate the performance using the \textit{Welfare Gap} relative to the Error-Free setting (the theoretical upper bound without the Dual Error, represented by the $y=0$ line). The compared settings are as follows:
% (1) \textbf{\textcolor[HTML]{9467BD}{Ante-Error-Free} (Practical Upper Bound):}
% An oracle seller that perfectly models the buyer's evaluation $h_i$ and best responds, submitting $p_i+\epsilon_{post}$.
% (2) \textbf{\textcolor[HTML]{C44E52}{Error-Naive} (Baseline):}
% Naive sellers that ignore the evaluation rule and report $p_i+\epsilon_{ante}$.
% (3) \textbf{\textcolor[HTML]{56A159}{EA-RAM} (Ours):}
% Sellers report according to belief $g_i$.

In this section, we run simulation experiments that explicitly control the magnitude of the Dual Error to evaluate EA-RAM's robustness.
We simulate $N=5$ heterogeneous sellers competing for one task with $V{=}20$ and $d{=}1.5$ (Table~\ref{tab:sim_settings}). To stress-test robustness, we use a highly competitive market with a small top-two surplus gap ($\Delta s=0.44$), where mild noise can flip the ranking. Performance is measured by the \textit{Welfare Gap} relative to the Error-Free upper bound, shown as the $y=0$ line. We compare:
% (1) \textbf{\textcolor[HTML]{9467BD}{Ante-Error-Free} (Practical Upper Bound):} an oracle seller that perfectly models $h_i$ and best responds by reporting $p_i+\epsilon_{post}$;
% (2) \textbf{\textcolor[HTML]{C44E52}{Error-Naive} (Baseline):} naive sellers that ignore evaluation and report $p_i+\epsilon_{ante}$;
% and
% (3) \textbf{\textcolor[HTML]{56A159}{EA-RAM} (Ours):} sellers report according to belief $g_i$.
(1) \textbf{\textcolor[HTML]{9467BD}{Ante-Error-Free} (Practical Upper Bound):} an oracle seller that perfectly models the buyer's acceptance probability $h_i$ and reports
accordingly;
(2) \textbf{\textcolor[HTML]{C44E52}{Error-Naive} (Baseline):} sellers that form predictions without accounting for evaluator-side distortion;
and
(3) \textbf{\textcolor[HTML]{56A159}{EA-RAM} (Ours):} sellers report according to their subjective acceptance belief $g_i$.

\paragraph{Error Robustness Analysis Results.}
\label{sec:simulation_results}

\begin{figure}[t]
    \centering
    \begin{subfigure}{0.48\textwidth}
        \centering
        \includegraphics[width=\linewidth]{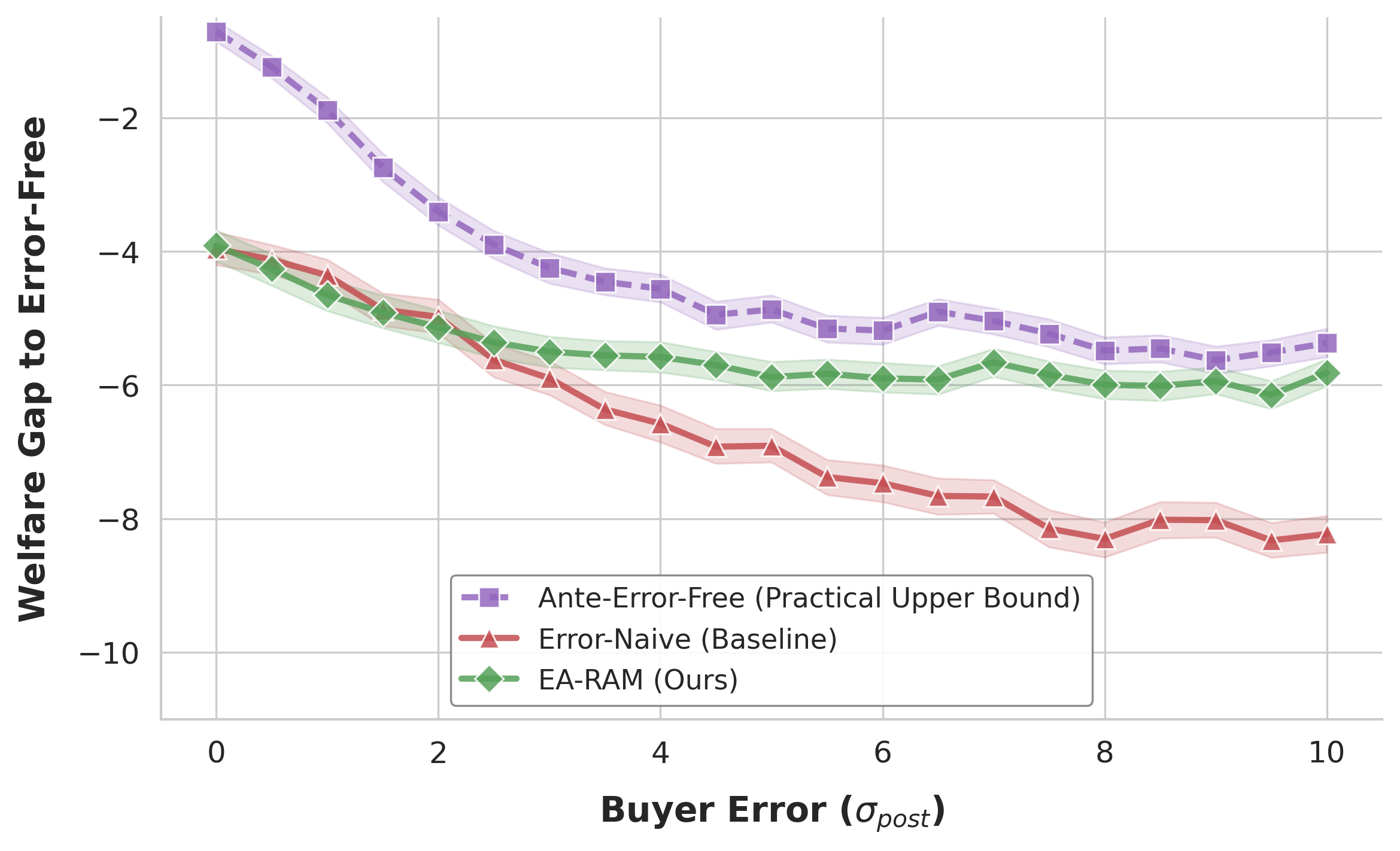}
        \vspace{-1.6em}
        \caption{Robustness to evaluation error $\sigma_{\text{post}}$. \textcolor[HTML]{56A159}{EA-RAM} closely tracks the \textcolor[HTML]{9467BD}{Ante-Error-Free} upper bound, while \textcolor[HTML]{C44E52}{Error-Naive} degrades sharply.}
        \label{fig:buyer_noise}
    \end{subfigure}
    \hfill
    \begin{subfigure}{0.48\textwidth}
        \centering
        \includegraphics[width=\linewidth]{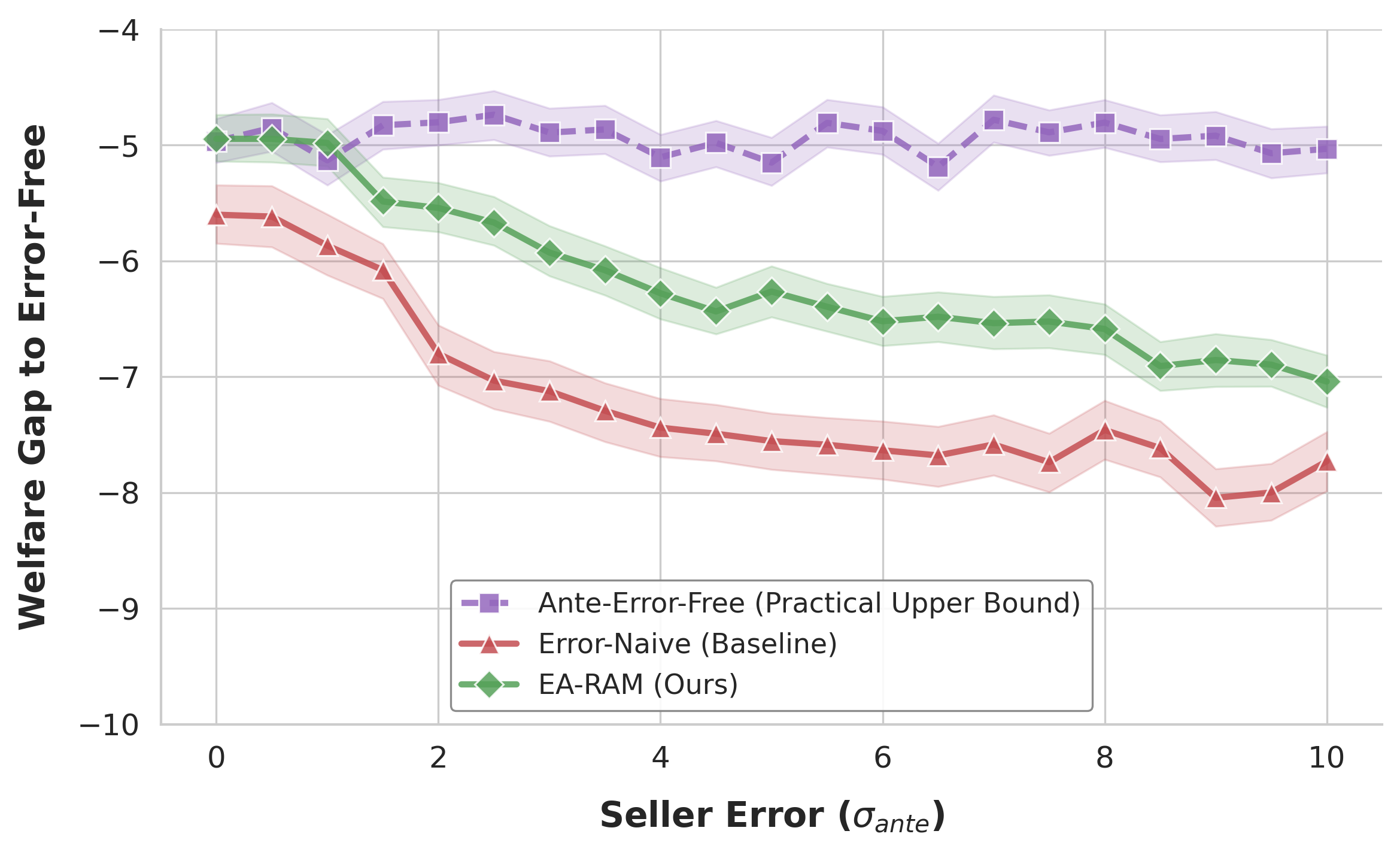}
        \vspace{-1.6em}
        \caption{Robustness to prediction error $\sigma_{\text{ante}}$. \textcolor[HTML]{56A159}{EA-RAM} consistently outperforms \textcolor[HTML]{C44E52}{Error-Naive} as seller information quality deteriorates.}
        \label{fig:seller_noise}
    \end{subfigure}
    \vspace{-0.6em}
    \caption{\textcolor[HTML]{56A159}{EA-RAM} remains stable under evaluator and prediction noise, avoiding the misallocation caused by naive reporting. Shaded regions show 95\% confidence intervals.}
    \label{fig:welfare_analysis}
    \vspace{-10pt}
\end{figure}

% Figure~\ref{fig:buyer_noise} evaluates robustness to evaluation error ($\sigma_{post}$).
% The \textcolor[HTML]{9467BD}{Ante-Error-Free} curve serves as an oracle upper bound (a perfectly informed seller best-responding to $h_i$).
% Although \textcolor[HTML]{56A159}{EA-RAM} cannot fully attain this bound due to estimation frictions, it closely tracks its trend, yielding a largely stable welfare gap as $\sigma_{post}$ increases.
% By contrast, the \textcolor[HTML]{C44E52}{Error-Naive} baseline deteriorates sharply and diverges from the bound.
% \textbf{Overall, \textcolor[HTML]{56A159}{EA-RAM} avoids the severe misallocation induced by ignoring the seller's strategic bidding.}
% Figure~\ref{fig:seller_noise} varies the prediction error ($\sigma_{ante}$) level of seller~0, controlling the degree of information contamination.
% \textbf{Although a decline in welfare is inevitable as seller information quality deteriorates, \textcolor[HTML]{56A159}{EA-RAM} consistently outperforms \textcolor[HTML]{C44E52}{Error-Naive} across the full range of $\sigma_{ante}$.}
% Internalizing the evaluation mechanism makes \textcolor[HTML]{56A159}{EA-RAM} degrade more gracefully, preserving the best efficiency feasible given the available imperfect information.

Figure~\ref{fig:buyer_noise} evaluates robustness to evaluation error ($\sigma_{post}$). 
The \textcolor[HTML]{9467BD}{Ante-Error-Free} curve is an upper bound, where sellers perfectly model $h_i$ and best respond. 
\textbf{\textcolor[HTML]{56A159}{EA-RAM} closely tracks this bound and maintains a stable welfare gap as $\sigma_{post}$ increases}, whereas \textcolor[HTML]{C44E52}{Error-Naive} deteriorates sharply. 
Figure~\ref{fig:seller_noise} varies seller~0's prediction error ($\sigma_{ante}$). 
Although welfare inevitably declines as seller information becomes noisier, \textbf{\textcolor[HTML]{56A159}{EA-RAM} consistently outperforms \textcolor[HTML]{C44E52}{Error-Naive}}, showing that internalizing the evaluation mechanism mitigates misallocation and makes the degradation milder.

\subsection{Real-World Experiments}

\subsubsection{Experimental Setup}
\label{sec:realworld-setup}

\paragraph{Benchmark and baselines.}
In this section, we benchmark EA-RAM against centralized baselines on real-world benchmarks.
We evaluate routing policies on RouterBench~\citep{hu2024routerbench}, which aggregates per-query accuracy and monetary cost for 11 LLMs (Appendix~\ref{app:dataset_details}), spanning both open-source (e.g., Llama, Mixtral) and proprietary models (e.g., GPT, Claude). 
We use queries from six benchmarks covering reasoning, coding, math, and knowledge-intensive understanding, and split them into train/test with a 70/30 ratio.
We compare EA-RAM against a diverse set of strong baselines: EmbedLLM~\citep{zhuang2025embedllm}, IRT-Router~\citep{song-etal-2025-irt}, RouteLLM~\citep{ong2025routellm}, FrugalGPT~\citep{chen2024frugalgpt}, and Cascade Routing~\citep{dekoninck2025a}. They cover a wide range of architectures, enabling a comprehensive comparison.

\paragraph{Implementation details.}

% EA-RAM trains two types of models. 
% Each LLM has a seller-side ex-ante predictor, trained separately to estimate the probability that it will correctly answer a query based on the query embedding. 
% The buyer maintains an ex-post evaluator that predicts answer correctness from concatenated query and answer embeddings. 
% Both modules are MLPs with two layers. The embedding model is \texttt{all-MiniLM-L6-v2}~\citep{sentence_transformers_all_minilm_l6_v2}. More details are provided in Appendix~\ref{app:other_implementation_details}.

EA-RAM trains two types of predictive modules.
Each LLM has a seller-side ex-ante predictor that estimates query-level
correctness from the query embedding, serving as an ex-ante proxy for
the evaluator's acceptance probability.
The buyer maintains an ex-post evaluator that estimates answer correctness
from concatenated query and answer embeddings.
Thus, both modules target answer correctness under different information sets:
the seller predicts it ex ante from the query alone, whereas the buyer evaluates
it ex post after observing the generated answer.
Both modules are two-layer MLPs. The embedding model is
\texttt{all-MiniLM-L6-v2}~\citep{sentence_transformers_all_minilm_l6_v2}.
More details are provided in Appendix~\ref{app:other_implementation_details}.

\begin{figure*}[t]
    \centering
    \includegraphics[width=\linewidth]{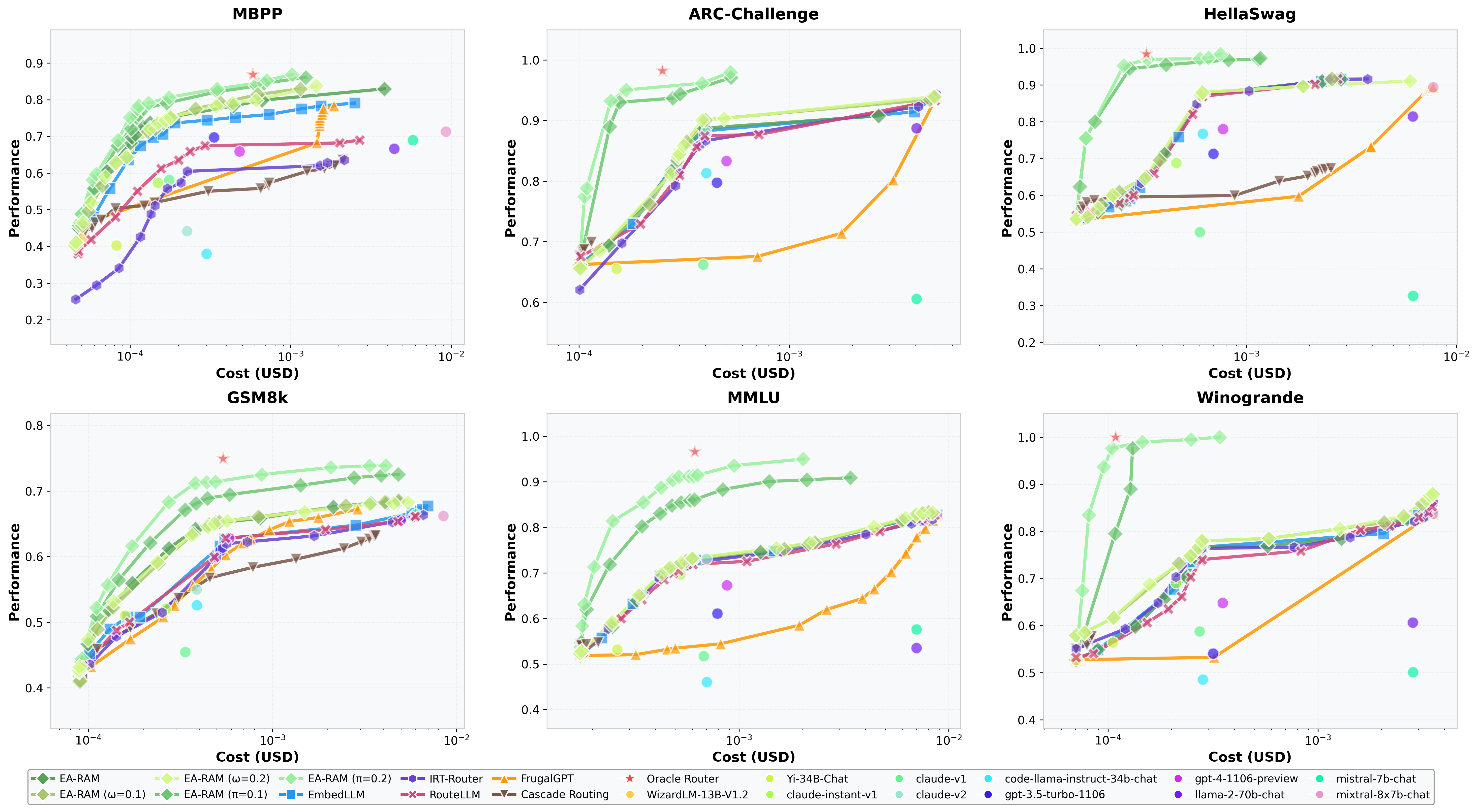}
    \vspace{-1.6em}
    \caption{Cost--performance Pareto frontiers on real-world benchmarks. As $\pi$ or $\omega$ increases, EA-RAM shifts the frontier upward and leftward, indicating a more favorable cost--quality trade-off than the centralized baselines.}
    \label{fig:all_datasets}
    \vspace{-1.6em}
\end{figure*}

\subsubsection{Routing Pareto Analysis}
\label{sec:realworld-results}

\begin{wraptable}{r}{0.54\linewidth}
\vspace{-3.0em}
\caption{Quantitative comparison (AIQ$\uparrow$). EA-RAM yields the best overall performance.}
\vspace{-0.5em}
\centering
\resizebox{\linewidth}{!}{
\setlength{\tabcolsep}{1.0pt}
\fontsize{7pt}{8pt}\selectfont
\begin{tabular}{l@{\hspace{0.2em}}cccccc}
\toprule
\textbf{Method} & \textbf{MBPP} & \textbf{ARC-C} & \textbf{HellaSwag} & \textbf{GSM8k} & \textbf{MMLU} & \textbf{Wino.} \\
\midrule
EmbedLLM & .77453 & .89855 & .89646 & .64596 & .78289 & .79154 \\
IRT-Router & .60465 & .89391 & .87664 & .63717 & .77604 & .79065 \\
RouteLLM & .67213 & .89498 & .85518 & .63841 & .77615 & .78825 \\
FrugalGPT & .70161 & .77691 & .72110 & .65646 & .67609 & .67387 \\
Cascade Routing & .58511 & .68921 & .62885 & .60751 & .54454 & .55891 \\
\cdashline{1-7}[1pt/1pt]
\textbf{EA-RAM} & \textbf{.80478} & \textbf{.90353} & \textbf{.89676} & \textbf{.67302} & \textbf{.78548} & \textbf{.79396} \\
with $\pi{=}0.1$ & .84932 & .96729 & .96919 & .71379 & .90021 & .97214 \\
with $\pi{=}0.2$ & .85792 & .97704 & .98029 & .73137 & .94303 & .99792 \\
with $\omega{=}0.1$ & .81834 & .91131 & .89938 & .67421 & .78616 & .81110 \\
with $\omega{=}0.2$ & .82138 & .91338 & .89912 & .67515 & .78899 & .81303 \\
\bottomrule
\end{tabular}
}
\label{tab:aiq_comparison}
\vspace{-1.2em}
\end{wraptable}

We evaluate EA-RAM on real-world benchmarks by sweeping each router's operating points to obtain cost--performance pairs $(c,a)$, where $c$ is the average cost per query and $a$ is empirical performance, and then extracting the Pareto frontier.
Beyond the base setting, EA-RAM uses two forms of seller-side local information. For oracle local information, we set
$\hat{p}^{(\pi)}_m(x)=(1-\pi)\hat{p}_m(x)+\pi y_m(x)$,
where $\pi\in[0,1]$ and $y_m(x)$ is the ground-truth label. For realistic local information, we retrieve the top-10 nearest training examples in the embedding space and use the similarity-weighted label average as a noisy provider-side signal $\tilde{y}_m(x)$, defining
$\hat{p}^{(\omega)}_m(x)=(1-\omega)\hat{p}_m(x)+\omega\tilde{y}_m(x)$,
where $\omega\in[0,1]$.
% Figure~\ref{fig:all_datasets} and Table~\ref{tab:aiq_comparison} show that EA-RAM achieves the best overall routing performance. Table~\ref{tab:aiq_comparison} reports \emph{Average Improvement in Quality} (AIQ; Appendix~\ref{app:aiq}), which averages Pareto-frontier quality over a shared cost range, with higher values indicating a better cost--performance trade-off.
% \textbf{EA-RAM already outperforms baselines in the base setting, and seller-side local information further shifts the frontier}: oracle information yields the largest gains, and realistic information also improves over the base setting across all datasets.
Figure~\ref{fig:all_datasets} and Table~\ref{tab:aiq_comparison} show that EA-RAM achieves a consistently favorable cost--performance trade-off. Table~\ref{tab:aiq_comparison} reports \emph{Average Improvement in Quality} (AIQ; Appendix~\ref{app:aiq}), which averages Pareto-frontier quality over a shared cost range.
In the base setting, \textbf{EA-RAM achieves the highest AIQ} across all six datasets. \textbf{More informative seller-side signals further improve the frontier}: oracle information provides an upper-bound sensitivity study, while retrieval-based local information yields consistent gains over the base setting.

\subsubsection{Robustness to Real-world Noise}

\begin{figure}[t]
    \centering
    \begin{subfigure}{0.49\textwidth}
        \centering
        \includegraphics[width=\linewidth]{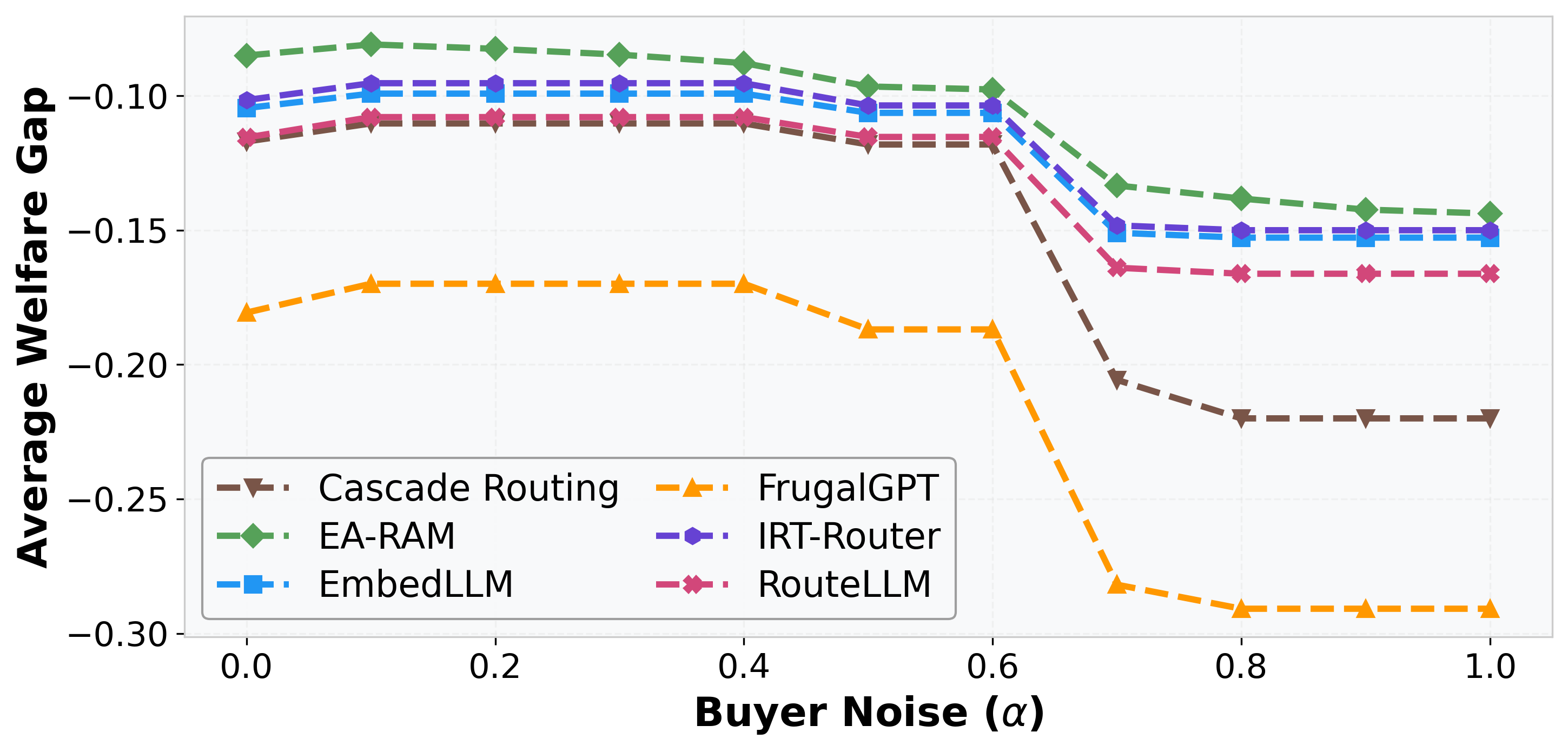}
        \vspace{-1.6em}
        \caption{Comparison on GSM8k.}
    \end{subfigure}
    \hfill
    \begin{subfigure}{0.49\textwidth}
        \centering
        \includegraphics[width=\linewidth]{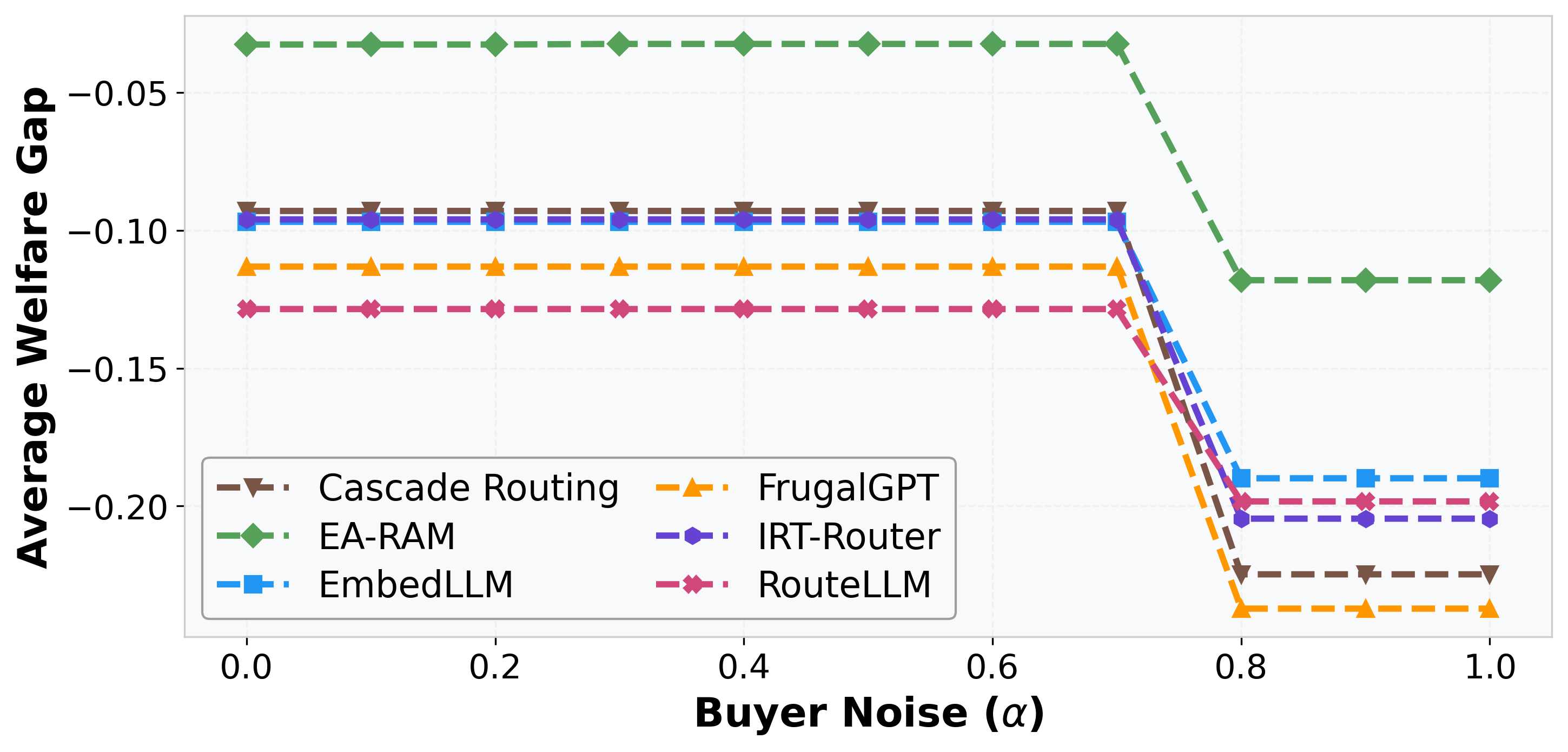}
        \vspace{-1.6em}
        \caption{Comparison on MBPP.}
    \end{subfigure}
    % \caption{Comparison on GSM8k and MBPP with a noisy LLM-as-a-Judge evaluator. We vary the judge weight $\alpha$ in the mixed evaluation score and report the average welfare gap relative to the error-free upper bound (higher is better). EA-RAM consistently achieves the smallest welfare loss across all evaluator-noise levels.}
    \vspace{-0.6em}
    \caption{Noisy LLM-as-a-Judge evaluation. Varying the judge weight $\alpha$, we report the average welfare gap to the error-free upper bound (higher is better). EA-RAM consistently incurs the smallest welfare loss.}
    \vspace{-1.6em}
    \label{fig:noisy_eval}
\end{figure}

% To further evaluate robustness to evaluator-side noise in a real benchmark, we conduct an additional experiment on GSM8k and MBPP using an LLM-as-a-judge evaluator (DeepSeek-V3.2). We define the evaluation score as $\text{score}=\alpha\cdot \text{score}_{\text{judge}}+(1-\alpha)\cdot \text{score}_{\text{groundtruth}}$, where $\alpha\in[0,1]$ controls the evaluator-noise level. We treat $\text{score}\ge 0.5$ as acceptance. Under this setup, all routing methods follow the same pipeline---select an LLM, evaluate the response, and compute welfare---and we report the average welfare gap relative to the error-free upper bound across all queries, where higher values indicate better performance.
% Figure~\ref{fig:noisy_eval} shows that \textbf{EA-RAM consistently achieves the best result across all values of $\alpha$, indicating the smallest welfare loss throughout.} This pattern remains stable from low to high evaluator noise, providing additional real-benchmark evidence that EA-RAM is robust to evaluator-side error.

To further test robustness to evaluator-side noise, we run a GSM8k/MBPP experiment with a LLM-as-a-judge evaluator (DeepSeek-V3.2). The evaluation score is
$\text{score}=\alpha\cdot \text{score}_{\text{judge}}+(1-\alpha)\cdot \text{score}_{\text{groundtruth}}$,
where $\alpha\in[0,1]$ controls evaluator noise, and $\text{score}\ge0.5$ indicates acceptance. All methods follow the same pipeline: select an LLM, evaluate its response, and compute welfare. We report the average welfare gap to the error-free upper bound over all queries, where higher is better. Figure~\ref{fig:noisy_eval} shows that \textbf{EA-RAM consistently performs best across all $\alpha$, incurring the smallest welfare loss}, providing real-benchmark evidence of robustness to evaluator-side error.

\subsubsection{Scalability and Efficiency}
\label{sec:scalability-efficiency}

\begin{wrapfigure}{r}{0.35\textwidth}
    \vspace{-1.0em}
    \centering
    \includegraphics[width=\linewidth]{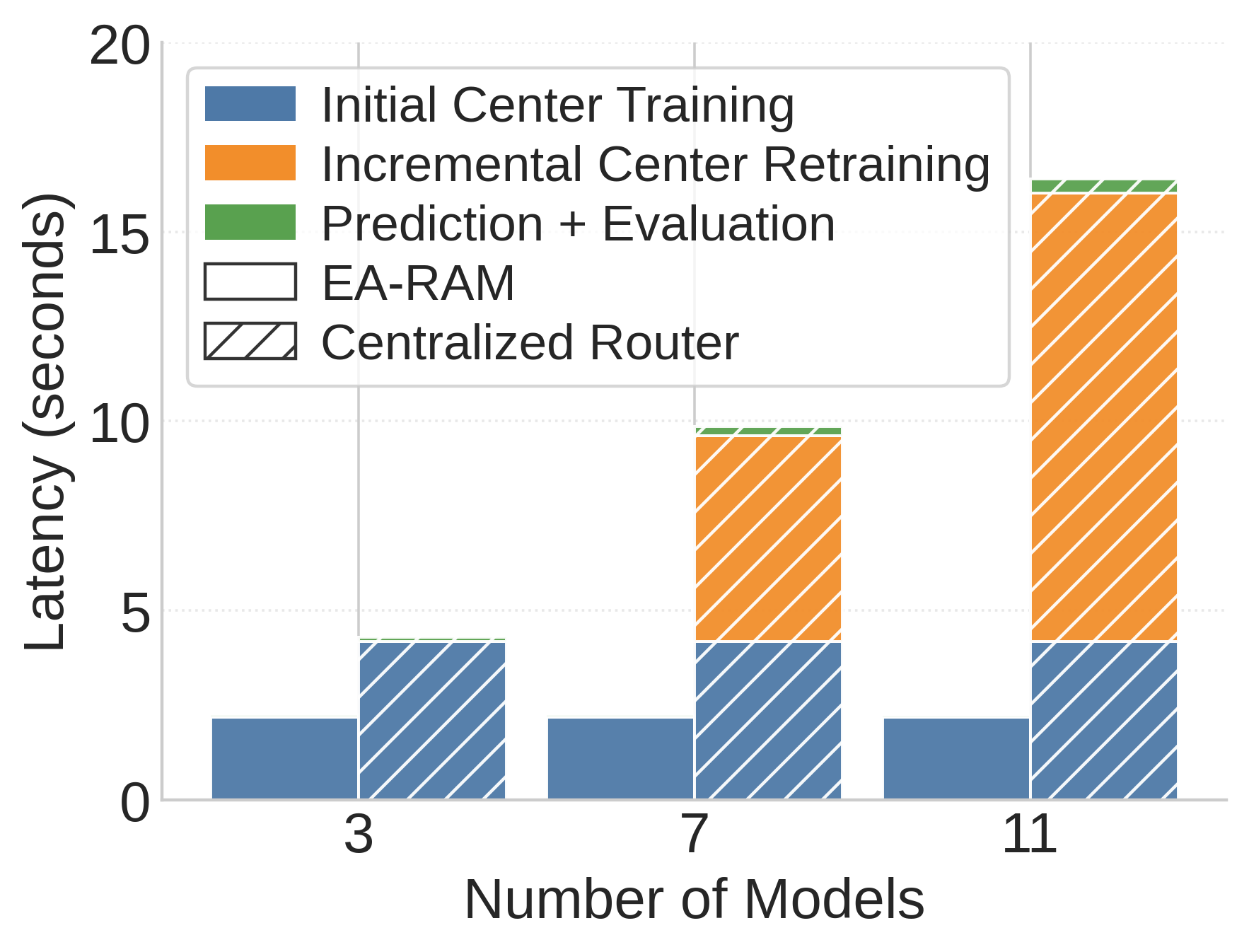}
    \vspace{-2.0em}
    \caption{Center-side latency.}
    \label{fig:latency_scaling}
    \vspace{-2.4em}
\end{wrapfigure}

To evaluate the practical scalability and efficiency of EA-RAM, we conduct an MBPP experiment as the number of candidate LLMs increases from 3 to 7 to 11, and examine both center-side computation latency and the extra communication overhead introduced by the reverse-auction.

% \par\WFclear
\paragraph{Center-side latency.}
Figure~\ref{fig:latency_scaling} shows that \textbf{the center-side latency of EA-RAM remains nearly constant as the model pool grows}, since the center only performs evaluator and auction computation and does not require retraining or per-model prediction. In contrast, the centralized router (EmbedLLM) becomes substantially slower as the model pool grows due to predictor retraining and ex-ante prediction. Detailed numerical results are provided in Appendix~\ref{app:detailed_efficiency}.

\par\WFclear

\begin{wraptable}{r}{0.47\linewidth}
\vspace{-1.2em}
\caption{Extra reverse-auction communication overhead on MBPP.}
\vspace{-0.6em}
\centering
\resizebox{\linewidth}{!}{
\setlength{\tabcolsep}{3.0pt}
\fontsize{8pt}{10pt}\selectfont
\begin{tabular}{cccc}
\toprule
$N$ &
\makecell[c]{Extra communication\\volume (bytes)} &
\makecell[c]{Computation\\latency (s)} &
\makecell[c]{Throughput\\threshold (MB/s)} \\
\midrule
3  & 44,643  & 0.044 & 0.338 \\
7  & 111,439 & 0.043 & 0.370 \\
11 & 177,377 & 0.042 & 0.384 \\
\bottomrule
\end{tabular}
}
\label{tab:comm_scaling}
\vspace{-0.4em}
\end{wraptable}

\paragraph{Communication overhead.}
% This efficiency gain comes at the cost of additional reverse-auction communication. For each query, EA-RAM sends the query to all $N$ providers and collects $N$ bids before selecting the winner, yielding an extra communication volume of $(N-1)S_q + NS_b$, where $S_q$ is the query size and $S_b$ is the bid size. Here we count only the overhead specific to the reverse-auction stage, excluding costs common to essentially all routing frameworks. Table~\ref{tab:comm_scaling} shows that this extra communication grows roughly linearly with $N$, whereas computation latency remains almost constant. Under parallel transmission with one dedicated channel per provider, communication becomes the bottleneck only when the effective per-channel throughput falls below about $0.338$--$0.384$ MB/s, suggesting that in our current text-based task setting, communication is less likely to dominate computation in typical high-bandwidth server environments.
The efficiency gain introduces extra reverse-auction communication: for each query, EA-RAM sends the query to all $N$ providers and collects $N$ bids, adding $(N-1)S_q+NS_b$ bytes, where $S_q$ and $S_b$ denote the query and bid sizes, respectively. This counts only auction-specific overhead, excluding costs shared by standard routing frameworks. As shown in Table~\ref{tab:comm_scaling}, communication grows roughly linearly with $N$, while computation latency stays nearly constant. With parallel provider channels, communication becomes the bottleneck only below $0.338$--$0.384$ MB/s per channel, suggesting that \textbf{it is less likely to dominate in typical high-bandwidth text-routing settings.}

%% file: section/conclusion.tex
% \section{Conclusion}
% We presented EA-RAM, an error-aware reverse-auction mechanism for LLM routing that replaces centralized capability estimation with provider-side ex-ante prediction and platform-side ex-post evaluation. This paradigm resolves the information--risk mismatch and removes the per-model profiling bottleneck by requiring only a model-agnostic evaluator.
% Theoretically, we model LLM routing as a Dual-Error environment with noisy provider predictions and noisy platform evaluations, and prove that EA-RAM remains incentive-aligned: it is Bayesian incentive compatible and individually rational, admits sufficient conditions for center rationality, and enjoys an explicit welfare-loss bound. Our analysis further reveals three robustness effects: opposite-signed errors can cancel, vanishing-tail links stabilize clear-cut cases via saturation, and additional noise smooths belief maps, weakly reducing the gains from marginal manipulation.
% Empirically, simulations and real-world benchmarks show that EA-RAM is robust to the Dual Error and achieves a better cost--performance frontier than strong centralized baselines, with additional gains when providers contribute local information. Overall, our results suggest that market-based routing offers a scalable and robust foundation for orchestrating heterogeneous LLM ecosystems under realistic uncertainty.

\section{Conclusion}
We presented EA-RAM, an error-aware reverse-auction mechanism for LLM routing that shifts capability estimation from a centralized router to providers while retaining a model-agnostic evaluator. Under Dual Error and the private-evaluation-belief structure, EA-RAM is DSIC in effective-surplus space, satisfies IR, admits sufficient conditions for CR, and enjoys bounded welfare loss. We further establish robustness properties including error compensation, saturation, and noise-induced smoothing. Experiments show robust cost--performance trade-offs and gains from provider-side local information, supporting market-based routing for heterogeneous LLM systems.

%% file: section/appendix.tex
\appendix

\section{More Implementation Details}
\label{app:implememtation_details}

\subsection{RouterBench}
\label{app:dataset_details}

The model pool of the selected part of the RouterBench dataset is:
\begin{itemize}[itemsep=0em, topsep=0em, leftmargin=1em]
    \item \textbf{Open Source Models:} Llama-70B-chat, Mixtral-8x7B-chat, Yi-34B-chat, Code Llama-34B, Mistral-7B-chat, and WizardLM-13B.
    \item \textbf{Proprietary Models:} GPT-4, GPT-3.5-turbo, Claude-instant-v1, Claude-v1, Claude-v2.
\end{itemize}

We use queries from six benchmarks covering reasoning (HellaSwag~\citep{zellers-etal-2019-hellaswag}, Winogrande~\citep{sakaguchi2021winogrande}, ARC-Challenge~\citep{clark2018think}), coding (MBPP~\citep{austin2021program}), math (GSM8k~\citep{cobbe2021training}), and knowledge-intensive understanding (MMLU~\citep{hendrycks2021measuring}).

\subsection{AIQ Computation}
\label{app:aiq}

Given a routing system family, we sample a set of parameterized routers and obtain a collection of cost--quality points
$\{(c_i, q_i)\}_{i=1}^n$, where $c_i$ is the average monetary cost per query and $q_i$ is the corresponding quality (e.g., accuracy / performance).
We then construct the \emph{non-decreasing convex hull} frontier $R_f(c)$ over a shared domain $[c_{\min}, c_{\max}]$.
AIQ is defined as the average quality of this frontier over the shared cost interval:
\begin{equation}
\mathrm{AIQ}(R_f)=\frac{1}{c_{\max}-c_{\min}}\int_{c_{\min}}^{c_{\max}} R_f(c)\, dc.
\label{eq:aiq-routerbench}
\end{equation}
In implementation, we evaluate the integral numerically by (i) aligning all methods to the same $[c_{\min},c_{\max}]$ via endpoint extrapolation when necessary, (ii) representing $R_f(c)$ with piecewise-linear segments, and (iii) computing the area under the curve using the trapezoidal rule, then normalizing by $(c_{\max}-c_{\min})$. 

For endpoint extrapolation, let a router have observed points spanning $[c_{\min}^R,c_{\max}^R]$ with corresponding frontier quality range $[q_{\min}^R,q_{\max}^R]$.
We left-extrapolate by extending the global minimum quality-cost point $(q_{\min}, c_{\min})$, and right-extrapolate by extending the router's maximum quality $q_{\max}^R$ to the global maximum cost $c_{\max}$ (i.e., $R_f(c)=q_{\max}^R$ for $c\in[c_{\max}^R,c_{\max}]$).

\subsection{Other Details}
\label{app:other_implementation_details}

EA-RAM's predictors and evaluator are trained with cross-entropy loss and AdamW for 100 epochs, with a batch size of 256 and a learning rate of $10^{-3}$. All experiments were done on a server with 8 NVIDIA GeForce RTX 3090 GPUs.

\section{More Experiments}

\subsection{Detailed Results for Efficiency Analysis}
\label{app:detailed_efficiency}
We provide detailed numerical results in Table~\ref{tab:latency_scaling}.

\begin{table*}[ht]
\centering
\small
\caption{Center-side latency under different numbers of candidate LLMs on MBPP. Lower is better.}
\label{tab:latency_scaling}
\setlength{\tabcolsep}{2.0pt}
\resizebox{\textwidth}{!}{
\begin{tabular}{lcccccc}
\toprule
Latency(s)/Method & EA-RAM (3) & EA-RAM (7) & EA-RAM (11) & Centralized Router (3) & Centralized Router (7) & Centralized Router (11) \\
\midrule
Initial Center Training & \textbf{2.184} & \textbf{2.184} & \textbf{2.184} & 4.181 & 4.181 & 4.181 \\
Incremental Center Retraining & -- & \textbf{0.000} & \textbf{0.000} & -- & 5.425 & 11.836 \\
Center Computation (Prediction + Evaluation) & \textbf{0.044} & \textbf{0.043} & \textbf{0.042} & 0.118 & 0.249 & 0.383 \\
Total Latency & \textbf{2.228} & \textbf{2.227} & \textbf{2.226} & 4.299 & 9.855 & 16.400 \\
\bottomrule
\end{tabular}}
\end{table*}

\section{Proof}\label{appendix:proof}

\subsection{Proof of Error-Free Setting's Property}
\label{appendix:proof_baseline}

\begin{proposition}[Optimality of the Error-Free Setting. Restatement of Proposition~\ref{prop:error_free_properties}]
\label{appen:prop:error_free_properties}
In the error-free setting, the mechanism satisfies all four desirable properties: DSIC, IR, CR, and EE.
\end{proposition}

\begin{proof}
\textbf{DSIC: }
Fix others' reports. Let $H=\max\{0, \max_{k\ne i}(\hat p_kV-\hat c_k)\}$.  
If $i$ wins, then $r_i=V\mu-H$, so $\E[r_i]=p_iV-H$ and
\begin{equation}
\E[U^{\text{seller}}_i\mid\text{win}] = \E[r_i]-c_i = p_iV-H-c_i = s_i-H.
\end{equation}
If $i$ loses, $U^{\text{seller}}_i=0$.  
Thus
\begin{equation}
\E[U^{\text{seller}}_i\mid \hat\theta_i] = (s_i-H)\cdot\Pr[\hat s_i\ge H],
\end{equation}
where the win probability depends on $\hat s_i$ relative to $H$.  

- If $s_i>H$, the best outcome is to win, yielding $s_i-H>0$. Truthful reporting sets $\hat s_i=s_i$, hence ensures winning.  
- If $s_i<H$, any win would yield $s_i-H<0$, strictly worse than losing. Truthful reporting sets $\hat s_i=s_i<H$, hence ensures losing.  
- If $s_i=H$, both win and loss yield expected utility $0$.  

Thus, truth-telling maximizes expected utility in all cases.

\textbf{IR: }
From DSIC, if $i$ wins then $\E[U^{\text{seller}}_i]=s_i-H\ge0$ because allocation requires $\hat s_i=s_i\ge H$. If $i$ loses then $\E[U^{\text{seller}}_i]=0$. Hence expected utility is never negative.

% \textbf{CR: }
% If seller $j$ is assigned, then
% \begin{equation}
% U^{\text{buyer}} = V\mu - r_j = V\mu - (V\mu - H) = H.
% \end{equation}
% Since the dummy seller has a score $s_0=0$, we have $H\ge 0$, so the buyer never loses. If no seller is assigned, then $U^{\text{buyer}}=0$.

\textbf{CR: }
In the error-free setting, $h_j=p_j$, so
$\E[r_j]=Vp_j-H$. Therefore,
$\E[U^{\mathrm{buyer}}]
=Vp_j-\E[r_j]=H\geq0$.
If the null allocation is selected, the buyer obtains zero utility.

\textbf{EE: }
With truth-telling, $\hat s_i=s_i$. The mechanism assigns only if $\max_i s_i>0$, and then to an $i$ maximizing $s_i$. This maximizes expected welfare $\max\{0,\max_i s_i\}$.
\end{proof}

\subsection{Strategic Properties}

\begin{theorem}[Dominant-Strategy Incentive Compatibility in Effective-Surplus Space. Restatement of Theorem~\ref{thm:bic}]
\label{appen:thm:bic}
Under the private-evaluation-belief structure above, for seller $i$ with effective-surplus type $\bar{T}_i = Vg_i-c_i$, reporting $\hat{s}_i=\bar{T}_i$ maximizes its subjective expected utility for every realization of the runner-up score $H$. Hence, truthful effective-surplus reporting is a weakly dominant strategy in the reduced-form mechanism.
\end{theorem}

\begin{proof}
Under the private-evaluation-belief structure, conditional on seller $i$'s information and a realized runner-up score $H$, its subjective expected payoff is $\bar{T}_i-H$ if it wins and $0$ otherwise. Reporting $\hat s_i=\bar T_i$ therefore wins exactly when
winning yields non-negative subjective expected utility. Any
$\hat s_i<\bar T_i$ may forgo a profitable allocation, whereas any
$\hat s_i>\bar T_i$ may induce an allocation with negative subjective
expected utility. Since this argument holds for every $H$,
$\hat s_i=\bar T_i$ is a weakly dominant strategy.
\end{proof}

% \begin{theorem}[Individual Rationality (IR). Restatement of Theorem~\ref{thm:IR}]
% \label{appen:thm:IR}
% At the equilibrium $\hat s=\bar T$, every seller satisfies IR:
% $\E[U^{\text{seller}}_i\mid \tilde\theta_i]\ge0$.
% \end{theorem}

% \begin{proof}
% \textbf{Win case.}
% Conditioned on $H\le \bar T_i$,
% \begin{equation}
% \label{eq:IR-win}
% \E[U^{\text{seller}}_i\mid \text{win},\tilde\theta_i]
% \;=\;
% \bar T_i-\E[H\mid H\le \bar T_i]\ \ge\ 0.
% \end{equation}
% Plugging $\hat s_i=\bar T_i$ into the interim utility $U^{\text{seller}}_i(\hat s_i)$ yields
% \begin{equation}
% \label{eq:IR-Ui}
% U^{\text{seller}}_i(\bar T_i)=F_H(\bar T_i)\big(\bar T_i-\E[H\mid H\le\bar T_i]\big)\ \ge\ 0.
% \end{equation}

% \textbf{Lose case.}
% If $H>\bar T_i$, then
% \begin{equation}
% \label{eq:IR-lose}
% \E[U^{\text{seller}}_i\mid \text{lose},\tilde\theta_i]=0.
% \end{equation}
% Combining \eqref{eq:IR-win}–\eqref{eq:IR-lose} gives $\E[U^{\text{seller}}_i\mid \tilde\theta_i]\ge0$.
% \end{proof}

\begin{theorem}[Individual Rationality (IR). Restatement of Theorem~\ref{thm:IR}]
\label{appen:thm:IR}
Under truthful effective-surplus reporting $\hat s_i=\bar T_i$, every
seller satisfies interim IR under its subjective belief:
$\E_i[U_i^{\mathrm{seller}}\mid\mathcal I_i]\geq0$.
\end{theorem}

\begin{proof}
Under truthful effective-surplus reporting, seller $i$'s interim subjective expected utility is
$\E_i[(\bar T_i-H)\mathbf{1}\{H\leq\bar T_i\}\mid\mathcal I_i]$.
Since $\bar T_i-H\geq0$ whenever $H\leq\bar T_i$, this expectation is non-negative. Hence,
$\E_i[U_i^{\mathrm{seller}}\mid\mathcal I_i]\geq0$.
\end{proof}

\begin{lemma}[Lipschitz Bound.]
\label{lem:link-lip}
The link function $\sigma$ is globally $L_\sigma$-Lipschitz, i.e., $|\sigma'(z)|\le L_\sigma$ for all $z\in\mathbb{R}$.
Therefore, for any real number $x$ and any random variable $\epsilon$ with $\E[\epsilon^2]<\infty$, the deviation between $\E[\sigma(x+\epsilon)]$ and $\sigma(x)$ satisfies
$\big|\E[\sigma(x+\epsilon)]-\sigma(x)\big| \le L_\sigma\,\E[|\epsilon|] \le L_\sigma\,\sqrt{\Var(\epsilon)+(\E[\epsilon])^2}$.
\end{lemma}

\begin{proof}
The proof consists of two steps: applying the Lipschitz property and then bounding the first absolute moment by the second moment.

\textbf{Step 1: Bounding the deviation via the Lipschitz constant.}
Since $\sigma$ is differentiable and satisfies $|\sigma'(z)|\leq L_\sigma$ for all $z$, the Mean Value Theorem implies that $\sigma$ is $L_\sigma$-Lipschitz. Hence, for $a=x+\epsilon$ and $b=x$,
\begin{equation}
|\sigma(x+\epsilon)-\sigma(x)|
\leq
L_\sigma|\epsilon|.
\label{eq:lip-pointwise}
\end{equation}
Taking expectations and applying Jensen's inequality gives
\begin{equation}
\begin{aligned}
\big|\E[\sigma(x+\epsilon)]-\sigma(x)\big|
&=
\big|\E[\sigma(x+\epsilon)-\sigma(x)]\big| \\
&\leq
\E\big[|\sigma(x+\epsilon)-\sigma(x)|\big] \\
&\leq
L_\sigma\E[|\epsilon|],
\end{aligned}
\label{eq:lip-expectation}
\end{equation}
where the last inequality follows from Eq.~\eqref{eq:lip-pointwise}.

\textbf{Step 2: Bounding the first absolute moment.}
By Jensen's inequality,
\begin{equation}
\E[|\epsilon|]
\leq
\sqrt{\E[\epsilon^2]}.
\label{eq:first-second-moment}
\end{equation}
Using $\E[\epsilon^2]=\Var(\epsilon)+(\E[\epsilon])^2$ in Eq.~\eqref{eq:first-second-moment}, we obtain
\begin{equation}
\E[|\epsilon|]
\leq
\sqrt{\Var(\epsilon)+(\E[\epsilon])^2}.
\label{eq:moment-bound}
\end{equation}
Combining Eq.~\eqref{eq:lip-expectation} and Eq.~\eqref{eq:moment-bound} yields
\begin{equation}
\big|\E[\sigma(x+\epsilon)]-\sigma(x)\big|
\leq
L_\sigma\E[|\epsilon|]
\leq
L_\sigma\sqrt{\Var(\epsilon)+(\E[\epsilon])^2},
\label{eq:lip-final}
\end{equation}
which proves the claim.
\end{proof}

\begin{proposition}[Approximate IC under Correlated Evaluation Beliefs. Restatement of Proposition~\ref{prop:appro_ic}]
\label{appen:prop:appro_ic}
Suppose seller $i$'s conditional evaluation belief satisfies
$|\E_i[\tilde\mu_i\mid H,\mathcal I_i]-g_i|\leq\delta_i$
for every runner-up score $H$ in its support. Then reporting
$\hat s_i=\bar T_i$ is $V\delta_i$-approximately dominant in the
reduced-form mechanism: for every alternative report and every $H$,
the conditional expected utility gain from deviating is at most
$V\delta_i$. Exact DSIC is recovered when $\delta_i=0$.
\end{proposition}

\begin{proof}
Let
$q_i(H)=\E_i[\tilde\mu_i\mid H,\mathcal I_i]$
and define
$\xi_i(H)=V(q_i(H)-g_i)$, so that
$|\xi_i(H)|\leq V\delta_i$.
Conditional on $H$, seller $i$'s expected payoff from winning is
$\bar T_i-H+\xi_i(H)$.

Reporting $\hat s_i=\bar T_i$ selects winning exactly when
$H\leq\bar T_i$. If an alternative report induces the same allocation,
the utilities coincide. If truthful reporting wins while the alternative
loses, the alternative can improve utility only when the winning payoff is
negative; since $\bar T_i-H\geq0$, this improvement is at most
$V\delta_i$. Conversely, if truthful reporting loses while the alternative
wins, then $\bar T_i-H<0$, so the alternative's gain is again at most
$V\delta_i$. Therefore truthful effective-surplus reporting is
$V\delta_i$-approximately dominant.
\end{proof}

\begin{theorem}[Center Rationality (CR). Restatement of Theorem~\ref{thm:CR-error-unified}]
\label{appen:thm:CR-error-unified}
Each of the following sufficient conditions guarantees CR:
(A) $\E[H] \ge V\Delta_{\text{gate}}$, where $\Delta_{\text{gate}} = L_\sigma \sqrt{b_{\text{post}}+a_{\text{post}}^2}$;
(B) $\Delta_{\text{cons}}=\E\!\left[p_{(1)}-h_{(1)}\right]\ge 0$.
\end{theorem}

\begin{proof}
If the null allocation is selected, the buyer obtains zero utility and CR holds trivially. We therefore consider a non-null selected seller, indexed by $(1)$, with runner-up score $H$.

Let $\mu_{(1)}\sim\mathrm{Bernoulli}(p_{(1)})$ denote the true fulfillment outcome and $\tilde\mu_{(1)}\sim\mathrm{Bernoulli}(h_{(1)})$ denote the evaluator's acceptance signal. Since the selected seller is paid
\begin{equation}
r_{(1)}=V\tilde\mu_{(1)}-H,
\label{eq:CR-payment}
\end{equation}
the buyer's realized utility is
\begin{equation}
U^{\text{buyer}}
=
V\mu_{(1)}-r_{(1)}
=
H+V\bigl(\mu_{(1)}-\tilde\mu_{(1)}\bigr).
\label{eq:CR-realized}
\end{equation}
Taking expectations gives
\begin{equation}
\E[U^{\text{buyer}}]
=
\E[H]
+
V\E\bigl[p_{(1)}-h_{(1)}\bigr].
\label{eq:Ubuyer-master}
\end{equation}

\noindent\textbf{(A)}
By the definition
$h_i=\E_{\varepsilon_{\mathrm{post}}}
[\sigma(\phi_i+\varepsilon_{\mathrm{post}})]$
and Lemma~\ref{lem:link-lip}, every seller $i$ satisfies
\begin{equation}
|h_i-p_i|
=
\left|
\E_{\varepsilon_{\mathrm{post}}}
[\sigma(\phi_i+\varepsilon_{\mathrm{post}})]
-
\sigma(\phi_i)
\right|
\le
L_\sigma
\sqrt{b_{\mathrm{post}}+a_{\mathrm{post}}^2}
=
\Delta_{\mathrm{gate}}.
\label{eq:CR-uniform}
\end{equation}
Since this bound holds uniformly for every seller, it also holds for the selected seller $(1)$. Hence,
\begin{equation}
\E\bigl[p_{(1)}-h_{(1)}\bigr]
\ge
-\Delta_{\mathrm{gate}}.
\label{eq:CR-selected-bound}
\end{equation}
Substituting \eqref{eq:CR-selected-bound} into \eqref{eq:Ubuyer-master} yields
\begin{equation}
\E[U^{\text{buyer}}]
\ge
\E[H]-V\Delta_{\mathrm{gate}}.
\label{eq:CR-gate-lb}
\end{equation}
Therefore, if $\E[H]\ge V\Delta_{\mathrm{gate}}$, then
$\E[U^{\text{buyer}}]\ge0$, establishing CR.

\noindent\textbf{(B)}
By the definition
$\Delta_{\mathrm{cons}}
=
\E[p_{(1)}-h_{(1)}]$,
Eq.~\eqref{eq:Ubuyer-master} gives
\begin{equation}
\E[U^{\text{buyer}}]
=
\E[H]
+
V\Delta_{\mathrm{cons}}.
\label{eq:CR-cons-eq}
\end{equation}
Because the dummy seller has score $0$, the runner-up score satisfies $H\ge0$, and hence $\E[H]\ge0$. Therefore, whenever $\Delta_{\mathrm{cons}}\ge0$, Eq.~\eqref{eq:CR-cons-eq} implies $\E[U^{\text{buyer}}]\ge0$, establishing CR.
\end{proof}

\begin{proposition}[Comparative Statics: Ability and Difficulty. Restatement of Proposition~\ref{prop:comparative_statics_unified}]
\label{appen:prop:comparative_statics_unified}
Holding fixed the distribution of $H$, seller $i$'s interim subjective
expected utility is weakly increasing in model ability $m_i$ and weakly
decreasing in task difficulty $d$.
\end{proposition}

\begin{proof}
Holding the distribution of $H$ fixed, seller $i$'s interim subjective
utility under truthful effective-surplus reporting is
$\E_i[(\bar T_i-H)_+\mid\mathcal I_i]$, where $\bar T_i=Vg_i-c_i$.
Since $\sigma$ is strictly increasing, $g_i$ is non-decreasing in
$\phi_i$. Moreover, $\phi_i$ is non-decreasing in $m_i$ and
non-increasing in $d$. Hence, $\bar T_i$ is weakly increasing in $m_i$
and weakly decreasing in $d$. Because $(T-H)_+$ is non-decreasing in
$T$ for every $H$, the claimed monotonicity follows.
\end{proof}

\subsection{Welfare Analysis}

% \begin{proposition}[Economic Efficiency (EE). Restatement of Proposition~\ref{prop:EE}]
% \label{appen:prop:EE}
% The equilibrium allocation induced by the error-aware setting attains the same expected welfare as the error-free benchmark if and only if the error-aware winner $i^\dagger$ is welfare-optimal.
% \end{proposition}

% \begin{proof}
% Let $i^\star$ denote an error-free welfare-maximizing seller. If seller $i$ is selected, then the realized welfare is
% \begin{equation}
% \label{eq:EE-realized-welfare}
% W_i = V\mu_i-c_i.
% \end{equation}
% Taking the expectation gives
% \begin{equation}
% \label{eq:EE-expected-welfare}
% \E[W_i] = V\E[\mu_i]-c_i = Vp_i-c_i.
% \end{equation}
% Hence, the expected welfare induced by the error-aware allocation is
% \begin{equation}
% \label{eq:EE-error-aware-welfare}
% \E[W_{i^\dagger}] = Vp_{i^\dagger}-c_{i^\dagger},
% \end{equation}
% whereas the error-free benchmark welfare is
% \begin{equation}
% \label{eq:EE-benchmark-welfare}
% \E[W_{i^\star}] = \max_i\{Vp_i-c_i\}.
% \end{equation}
% Therefore,
% \begin{equation}
% \label{eq:EE-iff}
% \E[W_{i^\dagger}] = \E[W_{i^\star}]
% \iff
% Vp_{i^\dagger}-c_{i^\dagger}=\max_i\{Vp_i-c_i\}.
% \end{equation}
% This proves the claim.
% \end{proof}

\begin{proposition}[Economic Efficiency (EE). Restatement of Proposition~\ref{prop:EE}]
\label{appen:prop:EE}
The equilibrium allocation induced by the error-aware setting attains the same expected welfare as the error-free benchmark if and only if the error-aware allocation index $i^\dagger$ is welfare-optimal.
\end{proposition}

\begin{proof}
Define $W_0=0$.
We adopt the convention
$p_0=g_0=h_0=c_0=0$ for the dummy seller.
Let $i^\star$ denote an error-free welfare-maximizing seller. If seller $i$ is selected, then the realized welfare is
\begin{equation}
\label{eq:EE-realized-welfare}
W_i = V\mu_i-c_i.
\end{equation}
Taking the expectation gives
\begin{equation}
\label{eq:EE-expected-welfare}
\E[W_i] = V\E[\mu_i]-c_i = Vp_i-c_i.
\end{equation}
Hence, the expected welfare induced by the error-aware allocation is
\begin{equation}
\label{eq:EE-error-aware-welfare}
\E[W_{i^\dagger}] = Vp_{i^\dagger}-c_{i^\dagger},
\end{equation}
whereas the error-free benchmark welfare is
\begin{equation}
\label{eq:EE-benchmark-welfare}
\E[W_{i^\star}] = \max_{i\in\bar{\mathcal I}}\{Vp_i-c_i\}.
\end{equation}
Therefore,
\begin{equation}
\label{eq:EE-iff}
\E[W_{i^\dagger}] = \E[W_{i^\star}]
\iff
Vp_{i^\dagger}-c_{i^\dagger}=\max_{i\in\bar{\mathcal I}}\{Vp_i-c_i\}.
\end{equation}
This proves the claim.
\end{proof}

\begin{theorem}[Welfare-loss bound. Restatement of Theorem~\ref{thm:welfare-loss-unified}]
\label{appen:thm:welfare-loss-unified}
Let $\bar{\mathcal I}=\mathcal I\cup\{0\}$ include the dummy seller representing the null allocation, with $p_0=g_0=h_0=c_0=0$. Let $i^\star\in\arg\max_{i\in\bar{\mathcal I}}\{Vp_i-c_i\}$ and $i^\dagger\in\arg\max_{i\in\bar{\mathcal I}}\{Vg_i-c_i\}$. Then the expected welfare loss satisfies $0\leq \E[W_{i^\star}]-\E[W_{i^\dagger}]\leq 2VL_\sigma(M_{\mathrm{post}}+M_{\mathrm{ante}})$.
\end{theorem}

\begin{proof}
\textbf{Step 1: Welfare comparison.}
Let $\E[W_i]=Vp_i-c_i$ and $\bar T_i=Vg_i-c_i$. Since $i^\dagger$ maximizes $\bar T_i$ over $\bar{\mathcal I}$,
\begin{equation}
\begin{aligned}
\E[W_{i^\star}]-\E[W_{i^\dagger}]
&=
\bigl(\E[W_{i^\star}]-\bar T_{i^\star}\bigr)
+
\bigl(\bar T_{i^\star}-\bar T_{i^\dagger}\bigr)
+
\bigl(\bar T_{i^\dagger}-\E[W_{i^\dagger}]\bigr) \\
&\le
\bigl(\E[W_{i^\star}]-\bar T_{i^\star}\bigr)
+
\bigl(\bar T_{i^\dagger}-\E[W_{i^\dagger}]\bigr).
\end{aligned}
\label{eq:key-sandwich}
\end{equation}
Moreover, since $i^\star$ maximizes the true expected welfare over
$\bar{\mathcal I}$, the welfare gap is non-negative. Therefore,
\begin{equation}
0
\le
\E[W_{i^\star}]-\E[W_{i^\dagger}]
\le
V\Bigl(
|g_{i^\star}-p_{i^\star}|
+
|g_{i^\dagger}-p_{i^\dagger}|
\Bigr).
\label{eq:two-errors}
\end{equation}

\textbf{Step 2: Bound the two error channels.}
By the definitions of the buyer's acceptance probability and the seller's subjective acceptance probability,
$h_i=\E_{\varepsilon_{\mathrm{post}}}
[\sigma(\phi_i+\varepsilon_{\mathrm{post}})]$
and
$g_i=\E_{\varepsilon_{\mathrm{post}},\varepsilon_{\mathrm{ante},i}}
[\sigma(\phi_i+\varepsilon_{\mathrm{post}}+\varepsilon_{\mathrm{ante},i})]$.
Hence,
\begin{equation}
|g_i-p_i|
\le
|g_i-h_i|+|h_i-p_i|.
\label{eq:dual-error-decomposition}
\end{equation}
Applying Lemma~\ref{lem:link-lip} to the evaluator-side error gives
\begin{equation}
|h_i-p_i|
\le
L_\sigma
\sqrt{b_{\mathrm{post}}+a_{\mathrm{post}}^2}
=
L_\sigma M_{\mathrm{post}}.
\label{eq:post-error-bound}
\end{equation}
Similarly, conditioning on $\varepsilon_{\mathrm{post}}$ and applying the same lemma to the seller-side belief error yields
\begin{equation}
|g_i-h_i|
\le
L_\sigma
\sqrt{b_{\mathrm{ante},i}+a_{\mathrm{ante},i}^2}
\le
L_\sigma M_{\mathrm{ante}}.
\label{eq:ante-error-bound}
\end{equation}
Combining \eqref{eq:dual-error-decomposition}--\eqref{eq:ante-error-bound},
\begin{equation}
|g_i-p_i|
\le
L_\sigma
\bigl(M_{\mathrm{post}}+M_{\mathrm{ante}}\bigr)
\qquad
\text{for all }i\in\bar{\mathcal I}.
\label{eq:gi-pi-uniform}
\end{equation}

\textbf{Step 3: Combine.}
Applying \eqref{eq:gi-pi-uniform} to $i^\star$ and $i^\dagger$ in
\eqref{eq:two-errors} gives
\begin{equation}
0
\le
\E[W_{i^\star}]-\E[W_{i^\dagger}]
\le
2V L_\sigma
\bigl(M_{\mathrm{post}}+M_{\mathrm{ante}}\bigr),
\label{eq:final-gap}
\end{equation}
which completes the proof.
\end{proof}

\subsection{Structural Insights}

\begin{proposition}[Opposite-Signed Error Compensation. Restatement of Proposition~\ref{prop:counteracting-error}]
\label{appen:prop:counteracting-error}
If $(g_i-h_i)(h_i-p_i)\le 0$ for $i\in\{i^\star,i^\dagger\}$, then the welfare-loss bound tightens to $2V L_\sigma\max\{M_{\mathrm{ante}},M_{\mathrm{post}}\}$.
\end{proposition}

\begin{proof}
% Let $i^\star \in \arg\max_j W_j$ denote the winner under the error-free setting, and let $i^\dagger \in \arg\max_j\{V g_j-c_j\}$ denote the winner selected by the error-involved mechanism.
Let $i^\star \in \arg\max_{j\in\bar{\mathcal I}}\{Vp_j-c_j\}$ denote the welfare-maximizing allocation under the error-free setting, and let $i^\dagger \in \arg\max_{j\in\bar{\mathcal I}}\{Vg_j-c_j\}$ denote the allocation selected by the error-involved mechanism.
We start from the welfare-gap decomposition established in Theorem~\ref{thm:welfare-loss-unified}:
\begin{equation}
\label{eq:opp-gap1}
0 \le \E[W_{i^\star}]-\E[W_{i^\dagger}] \le V\Big(|g_{i^\star}-p_{i^\star}|+|g_{i^\dagger}-p_{i^\dagger}|\Big).
\end{equation}
Consider the error decomposition $g_i-p_i = (g_i-h_i) + (h_i-p_i) \coloneqq a_i + b_i$.
The condition $(g_i-h_i)(h_i-p_i) \le 0$ implies that the two error components $a_i$ and $b_i$ have opposite signs. Consequently, their sum is bounded by the maximum of their absolute values:
\begin{equation}
\label{eq:opp-gap2}
|g_i-p_i| = |a_i+b_i| \le \max\{|a_i|,|b_i|\}.
\end{equation}
Applying this to the critical sellers $i \in \{i^\star,i^\dagger\}$ yields:
\begin{equation}
\label{eq:opp-gap3}
|g_{i^\star}-p_{i^\star}|+|g_{i^\dagger}-p_{i^\dagger}|
\ \le\ 
\max\{|a_{i^\star}|,|b_{i^\star}|\}
+\max\{|a_{i^\dagger}|,|b_{i^\dagger}|\}.
\end{equation}
Next, invoke the Lipschitz bounds from Lemma~\ref{lem:link-lip}:  
\begin{align}
|g_i-h_i|&\le L_\sigma\,M_{\mathrm{ante}},
\label{eq:opp-gap4a}\\
|h_i-p_i|&\le L_\sigma\,M_{\mathrm{post}},
\label{eq:opp-gap4b}
\end{align}
where the first arises because $g-h$ differs only by the ante channel, and the second because $h-p$ aggregates the post channels. Taking the supremum across sellers $j$ yields
\begin{equation}
\label{eq:opp-gap5}
\sup_j|g_j-h_j|\le L_\sigma M_{\mathrm{ante}},
\quad
\sup_j|h_j-p_j|\le L_\sigma M_{\mathrm{post}}.
\end{equation}
Combining \eqref{eq:opp-gap1}–\eqref{eq:opp-gap5}, we obtain
\begin{equation}
\label{eq:opp-gap6}
0\ \le\ \E[W_{i^\star}]-\E[W_{i^\dagger}]
\ \le\ 2V L_\sigma\,
\max\!\big\{M_{\mathrm{ante}},\ M_{\mathrm{post}}\big\}.
\end{equation}
Finally, because for any nonnegative $x,y$, $\max\{x,y\}<x+y$ when both are positive,  
we have
\begin{equation}
\label{eq:opp-gap7}
2V L_\sigma\,
\max\!\big\{M_{\mathrm{ante}},\ M_{\mathrm{post}}\big\}
\ <\
2V L_\sigma\,(M_{\mathrm{post}}+M_{\mathrm{ante}}),
\end{equation}
which completes the proof.
\end{proof}

\begin{proposition}[Saturation Robustness. Restatement of Proposition~\ref{prop:saturation-robustness}]
\label{appen:prop:saturation-robustness}
$\lim_{|x|\to\infty}\sigma'(x)=0$. For any error $\epsilon$ with $\E[\epsilon^2]<\infty$, the deviation $\Delta(\phi_i)=|\E[\sigma(\phi_i+\epsilon)]-\sigma(\phi_i)|$ satisfies $\lim_{|\phi_i|\to\infty}\Delta(\phi_i)=0$.
\end{proposition}

\begin{proof}
Using the integral form of the Mean Value Theorem, for any realization of $\epsilon$,
\begin{equation}
\sigma(\phi_i+\epsilon)-\sigma(\phi_i)
=\epsilon\int_0^1 \sigma'(\phi_i+t\epsilon)\,dt.
\end{equation}
Taking absolute values and expectations gives
\begin{equation}
\Delta(\phi_i)
\le \E\!\left[|\epsilon|\int_0^1 |\sigma'(\phi_i+t\epsilon)|\,dt\right].
\label{eq:delta-mvt-bound}
\end{equation}

Fix $\eta>0$. Since $\lim_{|x|\to\infty}\sigma'(x)=0$, there exists $R>0$ such that
\begin{equation}
|x|\ge R \ \Longrightarrow\ |\sigma'(x)|\le \eta.
\label{eq:tail-flat}
\end{equation}
Define the event $A=\{|\epsilon|\le |\phi_i|/2\}$ and split the expectation in \eqref{eq:delta-mvt-bound} over $A$ and $A^c$.

On $A$, for any $t\in[0,1]$,
\begin{equation}
|\phi_i+t\epsilon|
\ge |\phi_i|-t|\epsilon|
\ge |\phi_i|-\frac{|\phi_i|}{2}
=\frac{|\phi_i|}{2}.
\end{equation}
Hence, if $|\phi_i|\ge 2R$, then $|\phi_i+t\epsilon|\ge R$ and by \eqref{eq:tail-flat} we have $|\sigma'(\phi_i+t\epsilon)|\le \eta$. Therefore,
\begin{equation}
\E\!\left[|\epsilon|\int_0^1 |\sigma'(\phi_i+t\epsilon)|\,dt\;\mathbf 1_A\right]
\le \eta\,\E[|\epsilon|].
\label{eq:A-term}
\end{equation}

On $A^c$, we use boundedness of $\sigma'$:
\begin{equation}
\E\!\left[|\epsilon|\int_0^1 |\sigma'(\phi_i+t\epsilon)|\,dt\;\mathbf 1_{A^c}\right]
\le \|\sigma'\|_\infty\,\E\!\left[|\epsilon|\,\mathbf 1\{|\epsilon|>|\phi_i|/2\}\right].
\label{eq:Ac-term}
\end{equation}
Since $\E[\epsilon^2]<\infty$, we have $\E[|\epsilon|]<\infty$, and moreover
\begin{equation}
\E\!\left[|\epsilon|\,\mathbf 1\{|\epsilon|>t\}\right]\xrightarrow[t\to\infty]{}0
\end{equation}
(e.g., by dominated convergence with the dominating integrable random variable $|\epsilon|$).
Thus the right-hand side of \eqref{eq:Ac-term} vanishes as $|\phi_i|\to\infty$.

Combining \eqref{eq:delta-mvt-bound}--\eqref{eq:Ac-term} yields, for all sufficiently large $|\phi_i|$,
\begin{equation}
\Delta(\phi_i)
\le \eta\,\E[|\epsilon|]
+\|\sigma'\|_\infty\,\E\!\left[|\epsilon|\,\mathbf 1\{|\epsilon|>|\phi_i|/2\}\right].
\end{equation}
Taking $\limsup_{|\phi_i|\to\infty}$ gives $\limsup_{|\phi_i|\to\infty}\Delta(\phi_i)\le \eta\,\E[|\epsilon|]$. Because $\eta>0$ is arbitrary, we conclude that $\Delta(\phi_i)\to 0$ as $|\phi_i|\to\infty$.
\end{proof}

\begin{proposition}[Noise-Induced Flattening. Restatement of Proposition~\ref{prop:variance_flattening}]
\label{appen:prop:variance_flattening}
Let $\eta$ be an independent noise term with $\E[|\eta|]<\infty$. Define the perturbed belief maps by convolution:
$\tilde g_i(\phi)=\E[g_i(\phi+\eta)]$ and $\tilde h_i(\phi)=\E[h_i(\phi+\eta)]$.
If $g_i$ and $h_i$ are continuously differentiable with bounded derivatives, then
$\sup_{\phi}|\tfrac{\partial \tilde g_i}{\partial \phi}(\phi)|\le \sup_{\phi}|\tfrac{\partial g_i}{\partial \phi}(\phi)|$
and
$\sup_{\phi}|\tfrac{\partial \tilde h_i}{\partial \phi}(\phi)|\le \sup_{\phi}|\tfrac{\partial h_i}{\partial \phi}(\phi)|$.
Consequently, injecting additional independent noise weakly reduces the maximal sensitivity of both $\phi\mapsto g_i(\phi)$ and $\phi\mapsto h_i(\phi)$.
\end{proposition}

\begin{proof}
We prove the claim for $\tilde g_i$; the argument for $\tilde h_i$ is identical. By definition,
\begin{equation}
\tilde g_i(\phi)=\E[g_i(\phi+\eta)].
\end{equation}
Since $g_i$ is continuously differentiable and $g_i'$ is bounded, we may differentiate under the expectation (e.g., by dominated convergence) to obtain
\begin{equation}
\frac{\partial \tilde g_i}{\partial \phi}(\phi)=\E\!\left[g_i'(\phi+\eta)\right].
\end{equation}
Taking absolute values and using the bound $|g_i'(x)|\le \sup_{y}|g_i'(y)|$ for all $x$ yields
\begin{equation}
|\frac{\partial \tilde g_i}{\partial \phi}(\phi)|
=|\E[g_i'(\phi+\eta)]|
\le \E\!\left[|g_i'(\phi+\eta)|\right]
\le \sup_{y}|g_i'(y)|.
\end{equation}
Finally, taking the supremum over $\phi$ on the left-hand side gives
\begin{equation}
\sup_{\phi}|\tfrac{\partial \tilde g_i}{\partial \phi}(\phi)|\le \sup_{y}|g_i'(y)|,
\end{equation}
as desired. The same steps apply to $\tilde h_i$.
\end{proof}

\section{Limitations and Future Directions}
\label{app:limitation_future_direction}

Although this work, to our knowledge, is the first to introduce a reverse-auction paradigm for LLM routing and to explicitly model Dual Error in this setting, several limitations remain. Our present theory does not cover payment-rule variants such as bounded penalties or non-negative payments, which may alter the existing DSIC, IR, and CR guarantees. Although we provide robustness experiments with realistic noisy local information and noisy evaluators, our empirical evaluation still does not cover all practical settings. Moreover, while EA-RAM improves scalability by trading additional communication overhead for nearly fixed center-side computation, this advantage is most evident in our current text-based benchmarks, where payload sizes are relatively small; in communication-heavy settings such as image- or video-based tasks, or under degraded network conditions, communication may become a more significant bottleneck. In addition, the auction protocol requires broadcasting each query to all candidate providers before selection, which may expose user prompts to non-winning providers and raises security concerns. Finally, because routing and payment are tied to evaluator outcomes, providers may be incentivized to optimize toward the evaluator rather than the user's true underlying need.

These limitations suggest several directions for future work. An important extension is to develop communication-efficient reverse-auction mechanisms. Other promising directions include extending the theory to bounded-penalty or non-negative-payment variants, tightening the current welfare-loss bound, generalizing the framework beyond single-winner routing to support multi-provider cross-checking, cascading, and multi-turn settings, extending the model to heterogeneous risk preferences, and providing more principled support for non-binary evaluation tasks beyond the current threshold-based implementation. In addition, while EA-RAM is defined with respect to an announced task value $V$, it remains well defined as long as the buyer announces a reference value; a natural next step is to study more systematically how misspecification of $V$ changes participation incentives and the resulting operating point.

%% file: iclr2027_conference.bib
@article{porter2008fault,
  title={Fault tolerant mechanism design},
  author={Porter, Ryan and Ronen, Amir and Shoham, Yoav and Tennenholtz, Moshe},
  journal={Artificial Intelligence},
  volume={172},
  number={15},
  pages={1783--1799},
  year={2008},
  publisher={Elsevier}
}

@inproceedings{song-etal-2025-irt,
    title = "{IRT}-Router: Effective and Interpretable Multi-{LLM} Routing via Item Response Theory",
    author = "Song, Wei  and
      Huang, Zhenya  and
      Cheng, Cheng  and
      Gao, Weibo  and
      Xu, Bihan  and
      Zhao, GuanHao  and
      Wang, Fei  and
      Wu, Runze",
    editor = "Che, Wanxiang  and
      Nabende, Joyce  and
      Shutova, Ekaterina  and
      Pilehvar, Mohammad Taher",
    booktitle = "Proceedings of the 63rd Annual Meeting of the Association for Computational Linguistics (Volume 1: Long Papers)",
    month = jul,
    year = "2025",
    address = "Vienna, Austria",
    publisher = "Association for Computational Linguistics",
    url = "https://aclanthology.org/2025.acl-long.761/",
    doi = "10.18653/v1/2025.acl-long.761",
    pages = "15629--15644",
    ISBN = "979-8-89176-251-0"
}

@article{wang2025icl,
  title={ICL-Router: In-Context Learned Model Representations for LLM Routing},
  author={Wang, Chenxu and Li, Hao and Zhang, Yiqun and Chen, Linyao and Chen, Jianhao and Jian, Ping and Ye, Peng and Zhang, Qiaosheng and Hu, Shuyue},
  journal={arXiv preprint arXiv:2510.09719},
  year={2025}
}

@inproceedings{guo2024large,
  title={Large Language Model Based Multi-agents: A Survey of Progress and Challenges},
  author={Guo, Taicheng and Chen, Xiuying and Wang, Yaqi and Chang, Ruidi and Pei, Shichao and Chawla, Nitesh V and Wiest, Olaf and Zhang, Xiangliang},
  booktitle={IJCAI},
  year={2024}
}

@article{chen2026overview,
  title={An overview of domain-specific foundation model: key technologies, applications and challenges},
  author={Chen, Haolong and Chen, Hanzhi and Zhao, Zijian and Han, Kaifeng and Zhu, Guangxu and Zhao, Yichen and Du, Ying and Xu, Wei and Shi, Qingjiang},
  journal={Science China Information Sciences},
  volume={69},
  number={1},
  pages={111301},
  year={2026},
  publisher={Springer}
}

@inproceedings{
zhuang2025embedllm,
title={Embed{LLM}: Learning Compact Representations of Large Language Models},
author={Richard Zhuang and Tianhao Wu and Zhaojin Wen and Andrew Li and Jiantao Jiao and Kannan Ramchandran},
booktitle={The Thirteenth International Conference on Learning Representations},
year={2025},
url={https://openreview.net/forum?id=Fs9EabmQrJ}
}

@inproceedings{
ong2025routellm,
title={Route{LLM}: Learning to Route {LLM}s from Preference Data},
author={Isaac Ong and Amjad Almahairi and Vincent Wu and Wei-Lin Chiang and Tianhao Wu and Joseph E. Gonzalez and M Waleed Kadous and Ion Stoica},
booktitle={The Thirteenth International Conference on Learning Representations},
year={2025},
url={https://openreview.net/forum?id=8sSqNntaMr}
}

@article{
chen2024frugalgpt,
title={Frugal{GPT}: How to Use Large Language Models While Reducing Cost and Improving Performance},
author={Lingjiao Chen and Matei Zaharia and James Zou},
journal={Transactions on Machine Learning Research},
issn={2835-8856},
year={2024},
url={https://openreview.net/forum?id=cSimKw5p6R},
note={Featured Certification}
}

@inproceedings{
dekoninck2025a,
title={A Unified Approach to Routing and Cascading for {LLM}s},
author={Jasper Dekoninck and Maximilian Baader and Martin Vechev},
booktitle={Forty-second International Conference on Machine Learning},
year={2025},
url={https://openreview.net/forum?id=AAl89VNNy1}
}

@article{wang2017display,
  title={Display advertising with real-time bidding (RTB) and behavioural targeting},
  author={Wang, Jun and Zhang, Weinan and Yuan, Shuai and others},
  journal={Foundations and Trends{\textregistered} in Information Retrieval},
  volume={11},
  number={4-5},
  pages={297--435},
  year={2017},
  publisher={Now Publishers, Inc.}
}

@inproceedings{yuan2013real,
  title={Real-time bidding for online advertising: measurement and analysis},
  author={Yuan, Shuai and Wang, Jun and Zhao, Xiaoxue},
  booktitle={Proceedings of the seventh international workshop on data mining for online advertising},
  pages={1--8},
  year={2013}
}

@book{milgrom2004putting,
  title={Putting auction theory to work},
  author={Milgrom, Paul Robert},
  year={2004},
  publisher={Cambridge University Press}
}

@article{cramton1997fcc,
  title={The FCC spectrum auctions: An early assessment},
  author={Cramton, Peter},
  journal={Journal of Economics \& Management Strategy},
  volume={6},
  number={3},
  pages={431--495},
  year={1997},
  publisher={Wiley Online Library}
}

@article{takahashi2018strategic,
  title={Strategic design under uncertain evaluations: structural analysis of design-build auctions},
  author={Takahashi, Hidenori},
  journal={The RAND Journal of Economics},
  volume={49},
  number={3},
  pages={594--618},
  year={2018},
  publisher={Wiley Online Library}
}

@inproceedings{dubey2024auctions,
  title={Auctions with llm summaries},
  author={Dubey, Avinava and Feng, Zhe and Kidambi, Rahul and Mehta, Aranyak and Wang, Di},
  booktitle={Proceedings of the 30th ACM SIGKDD Conference on Knowledge Discovery and Data Mining},
  pages={713--722},
  year={2024}
}

@inproceedings{duetting2024mechanism,
  title={Mechanism design for large language models},
  author={Duetting, Paul and Mirrokni, Vahab and Paes Leme, Renato and Xu, Haifeng and Zuo, Song},
  booktitle={Proceedings of the ACM Web Conference 2024},
  pages={144--155},
  year={2024}
}

@article{hajiaghayi2024ad,
  title={Ad auctions for llms via retrieval augmented generation},
  author={Hajiaghayi, MohammadTaghi and Lahaie, S{\'e}bastien and Rezaei, Keivan and Shin, Suho},
  journal={Advances in Neural Information Processing Systems},
  volume={37},
  pages={18445--18480},
  year={2024}
}

@article{bhatt2025coalesce,
  title={Coalesce: Economic and security dynamics of skill-based task outsourcing among team of autonomous llm agents},
  author={Bhatt, Manish and Del Rosario, Ronald F and Narajala, Vineeth Sai and Habler, Idan},
  journal={arXiv preprint arXiv:2506.01900},
  year={2025}
}

@inproceedings{
hu2024routerbench,
title={RouterBench: A Benchmark for Multi-{LLM} Routing System},
author={Qitian Jason Hu and Jacob Bieker and Xiuyu Li and Nan Jiang and Benjamin Keigwin and Gaurav Ranganath and Kurt Keutzer and Shriyash Kaustubh Upadhyay},
booktitle={Agentic Markets Workshop at ICML 2024},
year={2024},
url={https://openreview.net/forum?id=IVXmV8Uxwh}
}

@inproceedings{
hendrycks2021measuring,
title={Measuring Massive Multitask Language Understanding},
author={Dan Hendrycks and Collin Burns and Steven Basart and Andy Zou and Mantas Mazeika and Dawn Song and Jacob Steinhardt},
booktitle={International Conference on Learning Representations},
year={2021},
url={https://openreview.net/forum?id=d7KBjmI3GmQ}
}

@inproceedings{zellers-etal-2019-hellaswag,
    title = "{H}ella{S}wag: Can a Machine Really Finish Your Sentence?",
    author = "Zellers, Rowan  and
      Holtzman, Ari  and
      Bisk, Yonatan  and
      Farhadi, Ali  and
      Choi, Yejin",
    editor = "Korhonen, Anna  and
      Traum, David  and
      M{\`a}rquez, Llu{\'i}s",
    booktitle = "Proceedings of the 57th Annual Meeting of the Association for Computational Linguistics",
    month = jul,
    year = "2019",
    address = "Florence, Italy",
    publisher = "Association for Computational Linguistics",
    url = "https://aclanthology.org/P19-1472/",
    doi = "10.18653/v1/P19-1472",
    pages = "4791--4800"
}

@article{cobbe2021training,
  title={Training verifiers to solve math word problems},
  author={Cobbe, Karl and Kosaraju, Vineet and Bavarian, Mohammad and Chen, Mark and Jun, Heewoo and Kaiser, Lukasz and Plappert, Matthias and Tworek, Jerry and Hilton, Jacob and Nakano, Reiichiro and others},
  journal={arXiv preprint arXiv:2110.14168},
  year={2021}
}

@article{clark2018think,
  title={Think you have solved question answering? try arc, the ai2 reasoning challenge},
  author={Clark, Peter and Cowhey, Isaac and Etzioni, Oren and Khot, Tushar and Sabharwal, Ashish and Schoenick, Carissa and Tafjord, Oyvind},
  journal={arXiv preprint arXiv:1803.05457},
  year={2018}
}

@article{austin2021program,
  title={Program synthesis with large language models},
  author={Austin, Jacob and Odena, Augustus and Nye, Maxwell and Bosma, Maarten and Michalewski, Henryk and Dohan, David and Jiang, Ellen and Cai, Carrie and Terry, Michael and Le, Quoc and others},
  journal={arXiv preprint arXiv:2108.07732},
  year={2021}
}

@article{sakaguchi2021winogrande,
  title={Winogrande: An adversarial winograd schema challenge at scale},
  author={Sakaguchi, Keisuke and Bras, Ronan Le and Bhagavatula, Chandra and Choi, Yejin},
  journal={Communications of the ACM},
  volume={64},
  number={9},
  pages={99--106},
  year={2021},
  publisher={ACM New York, NY, USA}
}

@misc{openrouter_auto_router,
  title        = {Auto Router},
  author       = {{OpenRouter}},
  howpublished = {\url{https://openrouter.ai/docs/guides/routing/routers/auto-router}},
  note         = {Accessed: 2026-05-06}
}

@misc{requesty_smart_routing,
  title        = {Smart LLM Routing},
  author       = {{Requesty}},
  howpublished = {\url{https://www.requesty.ai/solution/llm-routing}},
  note         = {Accessed: 2026-05-06}
}

@misc{sentence_transformers_all_minilm_l6_v2,
  title        = {all-MiniLM-L6-v2},
  author       = {{Sentence-Transformers}},
  howpublished = {\url{https://huggingface.co/sentence-transformers/all-MiniLM-L6-v2}},
  note         = {Accessed: 2026-05-06}
}

@article{zhang2026router,
  title={Router-r1: Teaching llms multi-round routing and aggregation via reinforcement learning},
  author={Zhang, Haozhen and Feng, Tao and You, Jiaxuan},
  journal={Advances in Neural Information Processing Systems},
  volume={38},
  pages={141233--141265},
  year={2026}
}

@inproceedings{zhang2026agentrouter,
  title={Agentrouter: A knowledge-graph-guided llm router for collaborative multi-agent question answering},
  author={Zhang, Zheyuan and Shi, Kaiwen and Yuan, Zhengqing and Wang, Zehong and Ma, Tianyi and Murugesan, Keerthiram and Galassi, Vincent and Zhang, Chuxu and Ye, Yanfang},
  booktitle={Proceedings of the 64th Annual Meeting of the Association for Computational Linguistics (Volume 1: Long Papers)},
  pages={788--809},
  year={2026}
}

@article{chen2026erskill,
  title={ERSkill: Evolving for Skill-Guided Adaptive Memory Retrieval},
  author={Chen, Haolong and Zhang, Liang and Li, Zhuo and Xue, Lei and Zhu, Guanrxu},
  journal={arXiv preprint arXiv:2608.12720},
  year={2026}
}

@article{li2026qoereasoner,
  title={QoEReasoner: An Agentic Reasoning Framework for Automated and Explainable QoE Diagnosis in RANs},
  author={Li, Qizhe and Chen, Haolong and Dai, Shan and Li, Zhuo and Hu, Zhiwei and Li, Xuan and Zhu, Guangxu and Shi, Qingjiang},
  journal={arXiv preprint arXiv:2606.01925},
  year={2026}
}

@article{xu2026cached,
  title={Cached model-as-a-resource: Provisioning large language model agents for edge intelligence in space-air-ground integrated networks},
  author={Xu, Minrui and Niyato, Dusit and Zhang, Hongliang and Kang, Jiawen and Xiong, Zehui and Mao, Shiwen and Han, Zhu},
  journal={IEEE Transactions on Networking},
  year={2026},
  publisher={IEEE}
}

@inproceedings{fan2026cost,
  title={Cost-effective communication: An auction-based method for language agent interaction},
  author={Fan, Yijia and Zhang, Jusheng and Cai, Kaitong and Yang, Jing and Tang, Chengpei and Wang, Jian and Wang, Keze},
  booktitle={Proceedings of the AAAI Conference on Artificial Intelligence},
  volume={40},
  number={35},
  pages={29412--29420},
  year={2026}
}

@article{habibi2025edge,
  title={Edge-LLM inference with cost-aware layer allocation and adaptive scheduling},
  author={Habibi, Sama and Ercetin, Ozgur},
  journal={IEEE Access},
  volume={13},
  pages={131614--131637},
  year={2025},
  publisher={IEEE}
}

@inproceedings{zhang2026mtrouter,
  title={MTRouter: Cost-Aware Multi-Turn LLM Routing with History--Model Joint Embeddings},
  author={Zhang, Yiqun and Li, Hao and Wang, Zihan and Feng, Shi and Yang, Xiaocui and Wang, Daling and Zhang, Bo and Bai, Lei and Hu, Shuyue},
  booktitle={Proceedings of the 64th Annual Meeting of the Association for Computational Linguistics (Volume 1: Long Papers)},
  pages={44206--44226},
  year={2026}
}

@inproceedings{niu2026routing,
  title={Routing with Generated Data: Annotation-Free LLM Skill Estimation and Expert Selection},
  author={Niu, Tianyi and Chen, Justin and Winata, Genta Indra and Zhang, Shi-Xiong and Chakraborty, Supriyo and Sahu, Sambit and Zhang, Yue and Stengel-Eskin, Elias and Bansal, Mohit},
  booktitle={Proceedings of the 64th Annual Meeting of the Association for Computational Linguistics (Volume 1: Long Papers)},
  pages={32441--32466},
  year={2026}
}

@article{groves1973incentives,
  title={Incentives in teams},
  author={Groves, Theodore},
  journal={Econometrica: Journal of the Econometric Society},
  pages={617--631},
  year={1973},
  publisher={JSTOR}
}
